\documentclass[journal]{IEEEtran}

\usepackage{multicol}
\usepackage{multirow}
\usepackage{subfigure}
\usepackage[dvipdf]{graphicx,color}
\usepackage{epsfig}
\usepackage{graphpap}
\usepackage{amsfonts}
\usepackage{amsmath}
\usepackage{amssymb}
\usepackage{enumerate}

\newtheorem{definition}{Definition}
\newtheorem{theorem}{Theorem}
\newtheorem{lemma}{Lemma}

\begin{document}

\title{Data-Driven Stochastic Models and Policies for Energy Harvesting Sensor Communications}

\author{\IEEEauthorblockN{Meng-Lin Ku, Yan Chen and K. J. Ray Liu} \thanks{M. L. Ku is affiliated with Department of Communication Engineering, National Central University, Taiwan; Y. Chen and K. J. R. Liu are affiliated with Department
of Electrical and Computer Engineering, University of Maryland, College Park, USA. }}

\maketitle

\begin{abstract}
Energy harvesting from the surroundings is a promising solution to perpetually power-up wireless sensor communications. This paper presents a data-driven approach of finding optimal transmission policies for a solar-powered sensor node that attempts to maximize net bit rates by adapting its transmission parameters, power levels and modulation types, to the changes of channel fading and battery recharge. We formulate this problem as a discounted Markov decision process (MDP) framework, whereby the energy harvesting process is stochastically quantized into several representative solar states with distinct energy arrivals and is totally driven by historical data records at a sensor node. With the observed solar irradiance at each time epoch, a mixed strategy is developed to compute the belief information of the underlying solar states for the choice of transmission parameters. In addition, a theoretical analysis is conducted for a simple on-off policy, in which a predetermined transmission parameter is utilized whenever a sensor node is active. We prove that such an optimal policy has a threshold structure with respect to battery states and evaluate the performance of an energy harvesting node by analyzing the expected net bit rate. The design framework is exemplified with real solar data records, and the results are useful in characterizing the interplay that occurs between energy harvesting and expenditure under various system configurations. Computer simulations show that the proposed policies significantly outperform other schemes with or without the knowledge of short-term energy harvesting and channel fading patterns.

\end{abstract}

\section{Introduction}
In traditional wireless sensor networks, sensor nodes are often powered by non-rechargeable batteries and distributed over a large area for data aggregation. But a major limitation of these untethered sensors is that the network lifetime is often dominated by finite battery capacity. Since the battery charge depletes with time, periodic battery or node replacement is required for prolonging the sensor node operations, though it becomes infeasible, costly and even impossible in some environments such as a large-scale network. As a result, there has been much research on designing efficient transmission mechanisms/protocols for saving energy in sensor communications \cite{C.Pandana05}.

Recently, energy harvesting has become an attractive alternative to circumvent this energy exhaustion problem by scavenging ambient energy sources (e.g., solar, wind, vibration, etc.) to replenish the sensors' power supply \cite{S.Sudevalayam11}. Though an inexhaustible energy supply from the environments enables wireless sensor nodes to function for a potentially infinite lifetime, management of the harvested energy remains a crucial issue due to the uncertainty of battery replenishment. In fact, most ambient sources occur randomly and sporadically in nature, and different sources exhibit different energy renewal processes in terms of predictability, controllability, and magnitude, requiring various design considerations for energy management.

In this paper, we focus on solar-powered wireless sensor networks, where each node is equipped with an energy harvesting device and a solar panel to collect surplus energy through the photovoltaic effect. Since the solar energy is uncontrollable and the rate of energy generation is typically small, the energy is temporarily stored and accumulated up to a certain amount in the rechargeable battery, which has limited storage capacity, for future data transmissions. But in practice, the amount of energy quanta available to a sensor could fluctuate dramatically even within a short period, and the level depends on many factors, such as the time of the day, the current weather, the seasonal weather patterns, the physical conditions of the environments around sensors, and the timescale (from seconds to days) of the energy management, to name but a few. This makes the prediction of energy harvesting conditions for the future intervals very challenging, even though the solar irradiance is partially predictable with the aid of daily irradiance patterns \cite{A.Kansal07}. Hence, there is a need for a stochastic energy harvesting model specific to each node, which is capable of capturing the dynamics of the solar energy associated with real data records. Besides, overly aggressive or conservative use of the harvested energy may either run out of the energy in the battery or fail to utilize the excess energy, resulting in the so-called energy outage or energy overflow problems, respectively. Consequently, another essential challenge lies in adaptively tuning the transmission parameters of sensor nodes in a smooth way that considers the randomness of energy generation and channel variation, avoids early energy depletion before the next management cycle, and maximizes certain utility functions through a finite or infinite horizon of epochs.

A wide variety of energy generation models have been adopted in the literature to study the performance of solar-powered sensor networks. In general, energy harvesting models can be categorized into two classes: deterministic models \cite{M.Tacca07}, \cite{S.Reddy10} and stochastic models \cite{A.Kansal07}, \cite{D.Niyato07}--\nocite{B.Medepally09}\nocite{N.Michelusi12}\nocite{N.Michelusi131}\nocite{A.Aprem13}\nocite{K.J.Prabuchandran13}\nocite{J.Lei09}\nocite{S.Mao12}\nocite{M.Kashef12}\nocite{Z.Wang12}\nocite{H.Li10}\cite{N.Michelusi13}. Deterministic models, which assume that energy arrival instants and amounts are known in advance by the transmitter, were applied in \cite{M.Tacca07} and \cite{S.Reddy10} for designing transmission schemes. The success of the energy management in this category rests on an assumption of accurate energy harvesting prediction over a somewhat long time horizon, whereas modeling mismatch occurs when the prediction interval is enlarged. Recently attention has shifted to stochastic models by accommodating the design of energy management to the randomness of energy renewal processes. The authors of \cite{D.Niyato07} developed an analytic solar radiation process with the parameters of cloud size and wind speed. By assuming that energy harvested in each time slot is identically and independently distributed, the energy generation process has been described via Bernoulli models with a fixed harvesting rate \cite{B.Medepally09}--\nocite{N.Michelusi12}\nocite{N.Michelusi131}\cite{A.Aprem13}. Other commonly used stochastic models that are uncorrelated across time include the uniform process \cite{A.Kansal07}, Poisson process \cite{K.J.Prabuchandran13}, and exponential process \cite{J.Lei09}. In some previous works, energy from ambient sources was modeled by a two-state Markov model to mimic the time-correlated harvesting behavior of a sensor node with time-slotted operation \cite{S.Mao12}--\nocite{M.Kashef12}\nocite{Z.Wang12}\nocite{H.Li10}\cite{N.Michelusi13}. However, there has been little research to validate the assumptions, along with exact physical interpretation, of the aforementioned stochastic models. It is essential to incorporate a data-driven stochastic model, which is capable of linking its underlying parameters to the dynamics of empirical energy harvesting data, into the design of sensor communications in order to develop more realistic performance characteristics.

Resource management for energy harvesting communications has been reported in the literature to optimize the system utility and to harmonize the energy consumption with the battery recharge rate. The optimization of energy usage is subject to a neutral constraint which stipulates that at each time instant, the energy expenditure cannot surpass the total amount of energy harvested so far. The utilities considered in previous works include the data throughput in \cite{S.Mao12}, \cite{M.Kashef12}, \cite{S.Yin13}--\nocite{P.S.Khairnar11}\nocite{O.Ozel11}\cite{N.Roseveare14}, the data queuing delay in \cite{T.Zhang13}, the packet error rate in \cite{A.Aprem13}, \cite{A.Seyedi10}, the transmission outage probability in \cite{D.Niyato07}, \cite{B.Medepally09}, \cite{S.Zhang13}, and the importance value of data packets in \cite{N.Michelusi131}, \cite{N.Michelusi13}. With deterministic energy and channel profiles, a utility maximization framework was investigated in \cite{M.Gorlatova13} to achieve smooth energy spending for a node. Applying a save-then-transmit protocol, the authors of \cite{S.Yin13} investigated a save-time ratio selection problem to maximize data throughput. In \cite{P.S.Khairnar11}, power and rate adaption were jointly designed for maximizing data throughput via an upper bound that is solely subject to an average power constraint. Directional water-filling was proposed in \cite{O.Ozel11} for throughput maximization. With stochastic models, the authors of \cite{J.Lei09} designed a threshold to decide whether to transmit or drop a message based on its importance. The outage probabilities of an energy harvesting node were analyzed in fading channels by taking into account both the energy harvesting and event arrival processes \cite{B.Medepally09}, \cite{S.Zhang13}. A simple power control policy was developed in \cite{Q.Wang13} to attain near optimal throughput in a finite-horizon case.

More recently, Markov decision processes (MDP) have been utilized to deal with the resource management problems for energy harvesting systems \cite{D.Niyato07}, \cite{A.Aprem13}, \nocite{K.J.Prabuchandran13}\cite{M.Kashef12}--\nocite{Z.Wang12}\cite{H.Li10}, \cite{T.Zhang13}, \cite{A.Seyedi10}\nocite{M.Gorlatova13}. When the battery replenishment, the wireless channel, and the packet arrival are regarded as Markov processes, sleep and wake-up strategies were developed in solar-powered sensor networks \cite{D.Niyato07}. Similar investigations were carried out with different reward functions, e.g., buffer delay \cite{M.Kashef12}, \cite{T.Zhang13}. In \cite{Z.Wang12} and \cite{A.Seyedi10}, transmission strategies for applications in body sensor networks and active networked tags were solved by casting them as MDP problems. In the presence of partial state information at transmitters, the problems of transmission scheduling and power control were addressed in \cite{A.Aprem13} and \cite{H.Li10}, respectively, by means of partially observable MDP. However, the aforementioned works all prearranged stochastic energy generation models for the development of transmission mechanisms without concern for the reality of the assumptions underlying the considered models. Further, none of these works linked the solar irradiance data, gathered by an energy harvesting node, to the constructions of the design frameworks and the optimal transmission policies.

In this paper, we present data-driven transmission policies for an energy harvesting source node that aims to transmit information packets to its sink over a wireless fading channel. For this we maximize the long-term bit rates of the communication link by adapting transmission power levels and modulation schemes to the source's knowledge of its current battery and channel status. We employ a Gaussian mixture hidden Markov model to quantify energy harvesting conditions into several representative solar states, whereby the underlying parameters enable us to effectively describe the statistical properties of the solar irradiance. To extract the underlying parameters, we use the learning ability of expectation-maximization (EM) algorithms to fit the hidden Markov model to the historical data record. Through the discretization, a stochastic discrete model that describes the generation of energy quanta is developed and integrated into our design frameworks to capture the interaction between the underlying parameters and the system parameters. The fading channel between the source and the sink is represented by a finite-state Markov model. The adaptive transmission problem is then formulated as a discounted MDP and solved by a value iteration algorithm. Both the energy wastage and the throughput degradation caused by data packet retransmission are taken into consideration when the average reward rate is maximized. Since the exact solar state is unknown to the energy harvesting sensor, an observation-based mixed strategy is developed to compute the belief state information and to decide the corresponding transmission parameters, based on the present measurement of the solar irradiance. In addition, we present a theoretical study on a simple on-off transmission policy to obtain more insight into our design framework. That means packets are transmitted at constant power and modulation levels if the action is ``ON", while no transmission occurs if the action is ``OFF". In this special case, there exists a threshold structure in the direction along the battery states, and the long-term expected bit rate is increased with the amount of energy quanta in the battery. Throughout this paper, a real data record of the solar irradiance measured by a solar site in Elizabeth City State University \cite{N.R.E.Laboratory12} is utilized to exemplify our design framework as well as its performance evaluation. The performance of the proposed transmission policies is validated by extensive computer simulations and compared with other radical policies with or without the knowledge of short-range energy harvesting and channel variation patterns.

The rest of this paper is organized as follows. In Section II, we describe the stochastic energy harvesting model, the training of its underlying parameters, and its connection to the real data record. The MDP formulation of the adaptive transmission is presented in Section III, followed by the optimization of the policies and the mixed strategy in Section IV. Section V is devoted to the analysis of the threshold structure for a simple on-off transmission policy. Simulation results are presented in Section VI, and concluding remarks are provided in Section VII.

\section{Stochastic Energy Harvesting Models and Training}
The model for describing the harvested energy depends on various parameters, such as weather conditions (e.g., sunny, cloudy, rainy), sunshine duration (e.g, day and night), and behavior of the rechargeable battery (e.g., storage capacity). Besides, the solar energy usually evolves in a smooth fashion over a short time period. We focus on modeling the solar power from the measurements by using a hidden Markov chain, and establish a framework to extract the underlying parameters that can characterize the availability of solar power.

We begin with a toy example to justify the rationality of the proposed energy harvesting models. Consider a real data record of irradiance (i.e., the intensity of the solar radiation in units $\mu$W/cm$^2$) for the month of June from 2008 to 2010, measured by a solar site in Elizabeth City State University, with the measurements taken at five-minute intervals \cite{N.R.E.Laboratory12}. In Fig. \ref{Toy_example_time_series}, the time series of the irradiance is sketched over twenty-four hours for June 15$^{th}$, 2010, along with the average results for the month of June in 2008 and 2010. We can make the following observations from this figure. First, the daily solar radiation fluctuates slowly within a short time interval, but could suddenly change from the current level to adjacent levels with higher or lower mean values. Second, the average irradiance value is sufficiently high only from the early morning (seven o'clock) to the late afternoon (seventeen o'clock). We refer to this time duration as the sunlight active region. Third, the evolution of the diurnal irradiance follows a very similar time-symmetric mask, whereas the short-term profiles of different days can be very different and unpredictable. By considering the irradiance from seven o'clock to seventeen o'clock for June in 2008, 2009 and 2010, Fig. \ref{Histogram_irradiance} shows the corresponding histogram plotted against the irradiance on the x-axis, which represents the percentage of the occurrences of data samples in each bin of width $10^3$ $\mu$W/cm$^2$. It can be seen that the irradiance behaves like a mixture random variable generated by a number of normal distributions. In fact, the solar radiation incident on an energy harvesting device is affected by its surrounding obstacles (e.g., cloud and terrain), which yield absorption, reflection and scattering phenomena, and it can be intuitively presented as a Gaussian random variable by the law of large numbers. These observations motivate us to describe the evolution of the irradiance via a hidden Markov chain with a finite number of possible states, each of which is specified by a normal distribution with unknown mean and variance.

\begin{figure}[t]
  \centering
  \subfigure[]{
    \label{Toy_example_time_series} 
    \includegraphics[width=0.23\textwidth]{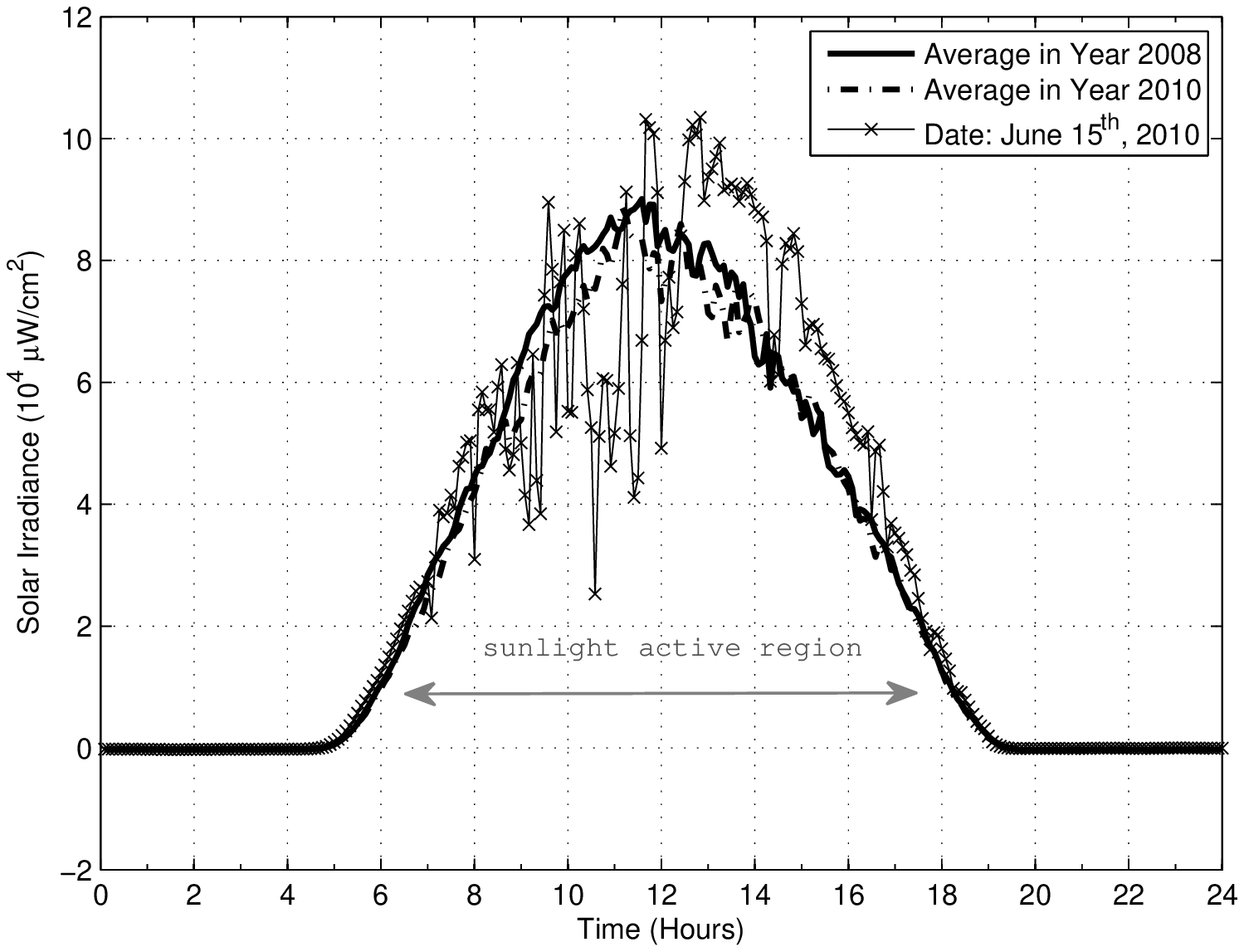}}
  \hspace{0.0in}
  \subfigure[]{
    \label{Histogram_irradiance} 
    \includegraphics[width=0.23\textwidth]{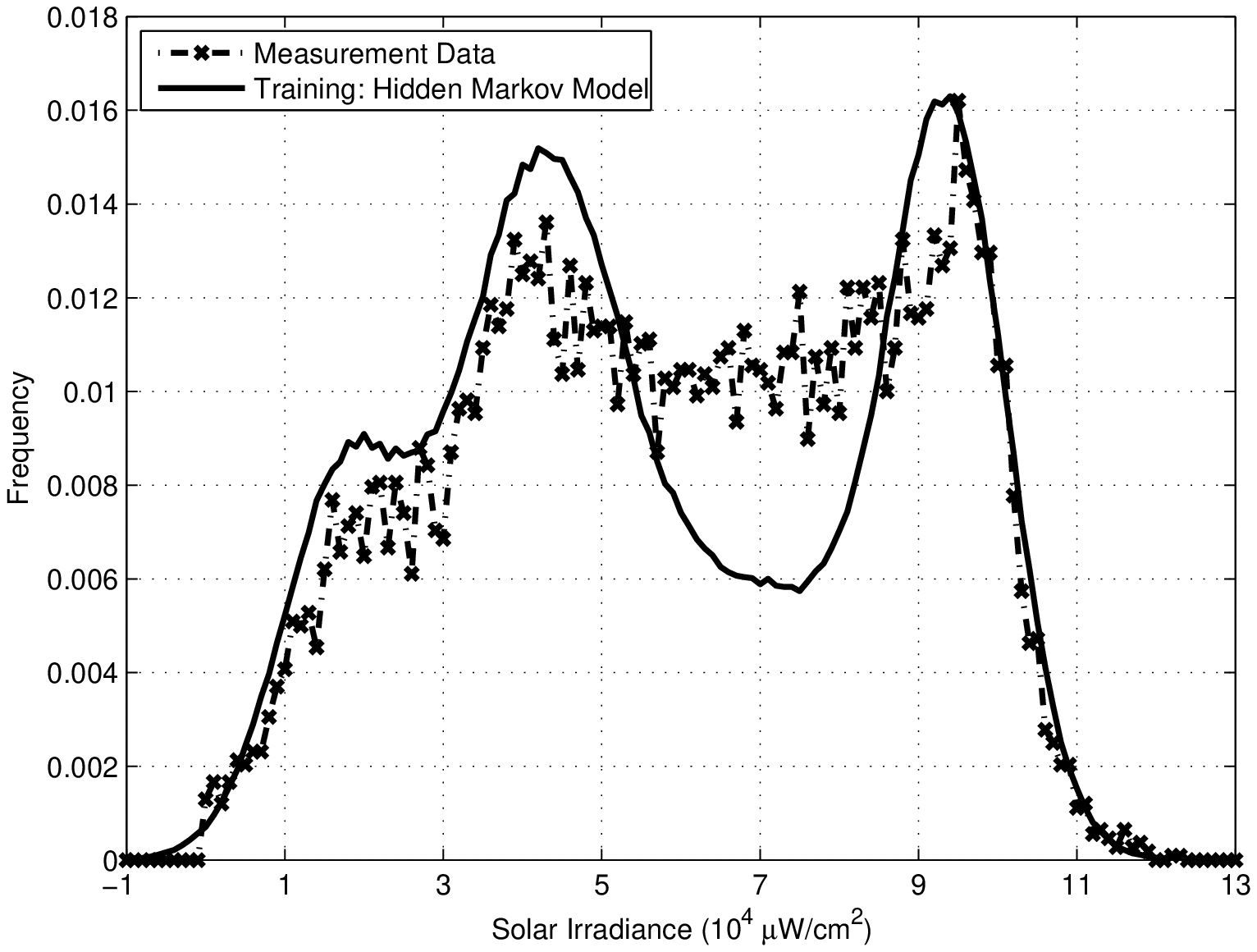}}
  \caption{Toy examples of solar irradiance measured by a solar site in Elizabeth City State University. (a) Time series of the daily irradiance in June. (b) Histogram of the irradiance during a time period of seven o'clock to seventeen o'clock for the month of June from 2008 to 2010.}
  \label{fig:subfig} 
\end{figure}

\begin{figure}[h]
\centering
\includegraphics[width=0.4\textwidth]{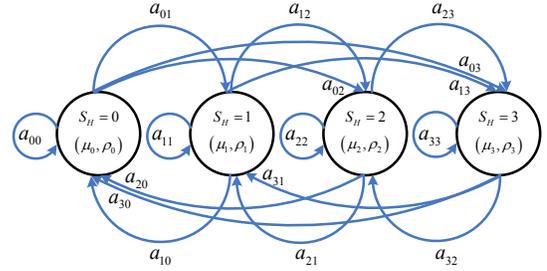}
\caption{Gaussian mixture hidden Markov chain of the solar power harvesting model with the underlying parameters $(\mu_j, \rho_j)$ ($N_H=4$).}
\label{Solar_power_state}
\end{figure}

An $N_H$-state solar power harvesting hidden Markov model is illustrated in Fig. \ref{Solar_power_state}, where the underlying normal distribution for the $j^{th}$ state is specified by the parameters of the mean $\mu_j$ and the variance $\rho_j$. The solar irradiance can be classified into several states $S_H$ to represent energy harvesting conditions such as `Excellent', `Good', `Fair', `Poor', etc. Without loss of generality, the solar states are numbered in ascending order of the mean values of the underlying parameters $\mu_j$. Let $S_H^{(t)}$ be the solar state at time instant $t$. We further assume that the hidden Markov model is time homogeneous and governed by the state transition probability $P\left( {\left. {S_H^{\left( t \right)}  = j} \right|S_H^{\left( {t - 1} \right)}  = i} \right) = a_{ij} $, for $i, j= 0,\ldots, N_H-1$. Let ${\bf x} = \left\{ {X^{\left( 1 \right)}  = x_1 , \ldots ,X^{\left( T \right)}  = x_T } \right\}$ be a sequence of observed data over a measurement period of $T$, corresponding to a sequence of hidden states ${\bf s} = \left\{ {S_H^{\left( 1 \right)}  = s_1 , \ldots ,S_H^{\left( T \right)}  = s_T } \right\}$. The parameters of the model are thus defined as ${\bf \mbox{\boldmath{$\Theta $}} } = \left\{ {{\mbox{\boldmath{$\mu $}} },{\mbox{\boldmath{$\rho $}} },{\bf a}} \right\}$, where ${\mbox{\boldmath{$\mu $}} } = \left[ {\mu_0 , \ldots ,\mu _{(N_H-1)} } \right]^T $, ${\mbox{\boldmath{$\rho $}}
 } = \left[ {\rho_0 , \ldots ,\rho _{(N_H-1)} } \right]^T $, and ${\bf a} = \left[ {a_{00} ,a_{01} , \ldots ,a_{(N_H-1)(N_H-1)} } \right]^T$. The probabilistic model can be trained by an EM algorithm, which is a general method of finding the maximum-likelihood (ML) estimate for the state parameters of underlying distributions from incomplete observed data, as follows \cite{J.A.Bilmes98}:
 \begin{align}
 \label{EM_formula}
    {\bf \Theta }^{\left( n \right)} &= \arg \mathop {\max }\limits_{\bf \Theta } \mathbb{E}_{\bf s} \left[ {\left. {\log P\left( {\left. {{\bf x},{\bf s}} \right|{\bf \Theta }} \right)} \right|{\bf x},{\bf \Theta }^{\left( {n - 1} \right)} } \right]\\
     &= \arg \mathop {\max }\limits_{\bf \Theta } \sum\nolimits_{\bf s} {\log P\left( {\left. {{\bf x},{\bf s}} \right|{\bf \Theta }} \right)}\cdot P\left( {\left. {{\bf x},{\bf s}} \right|{\bf \Theta }^{\left( {n - 1} \right)} } \right)\,, \nonumber
\end{align}
where ${\bf \Theta }^{\left( n \right)}$ is the estimation update at the $n^{th}$ iteration. For notational convenience, the following probabilities are defined:\begin{align}
\label{Alpha_term}
 \alpha _i^{\left( t \right)}  = P\left( {\left. {X^{\left( 1 \right)}  = x_1 , \ldots ,X^{\left( t \right)}  = x_t ,S_H^{\left( t \right)}  = i } \right|{\bf \Theta }^{\left( {n - 1} \right)} } \right)\, ;
\end{align}\begin{align}
\label{Beta_term}
 \beta _i^{\left( t \right)}  = P\left( {\left. {X^{\left( {t + 1} \right)}  = x_{t + 1} , \ldots ,X^{\left( T \right)}  = x_T } \right|S_H^{\left( t \right)}  = i ,{\bf \Theta }^{\left( {n - 1} \right)} } \right)\, .
\end{align}
From (\ref{Alpha_term}) and (\ref{Beta_term}), it further gives
\begin{align}
\label{Gamma_term}
 \gamma _i^{\left( t \right)}  = P\left( {\left. {S_H^{\left( t \right)}  = i } \right|{\bf x},{\bf \Theta }^{\left( {n - 1} \right)} } \right) = \frac{{\alpha _i^{\left( t \right)} \beta _i^{\left( t \right)} }}{{\sum\nolimits_{i' = 0}^{N_H-1} {\alpha _{i'}^{\left( t \right)} \beta _{i'}^{\left( t \right)} } }}\, ;
\end{align}
\begin{align}
\label{xi_term}
 \xi _{ij}^{\left( t \right)}  &= P\left( {\left. {{\bf x}, S_H^{\left( t \right)}  = i ,S_H^{\left( {t+1} \right)}  = j} \right|{\bf \Theta }^{\left( {n - 1} \right)} } \right) \\
 &= \frac{{\alpha _i^{\left( t \right)} a_{ij}^{\left( {n - 1} \right)} \beta _j^{\left( {t+1} \right)} f_j^{(n-1)} \left( {x_{t+1} } \right) }}{{\sum\nolimits_{i' = 0}^{N_H-1} {\sum\nolimits_{j' = 0}^{N_H-1} {\alpha _{i'}^{\left( t \right)} a_{i'j'}^{\left( {n - 1} \right)} \beta _{j'}^{\left( {t+1} \right)} f_{j'}^{(n-1)} \left( {x_{t+1} } \right) } } }}\, , \nonumber
\end{align}where $f_j^{(n-1)} \left( {x } \right) = \mathcal{N}\left( {x;\mu _j^{\left( {n - 1} \right)} ,\rho _j^{\left( {n - 1} \right)} } \right)$ represents the normal distribution, and the relationship among $\gamma _i^{\left( t \right)}$, $\alpha _i^{\left( t \right)}$, and $\beta _i^{\left( t \right)} $ in (\ref{Gamma_term}) is due to the conditional independence of the Markov chain: \begin{align}
\label{Conditional_indep_Markov}
 &P\left( {\left. {\bf x}, {S_H^{\left( t \right)}  = i } \right|{\bf \Theta }^{\left( {n - 1} \right)} } \right)  \\
 &=P\left( {\left. {X^{\left( 1 \right)}  = x_1 , \ldots ,X^{\left( t \right)}= x_t ,S_H^{\left( t \right)}  = i } \right|{\bf \Theta }^{\left( {n - 1} \right)} } \right) \nonumber\\
 &\;\;\;\cdot P\left( {\left. {X^{\left( {t + 1} \right)}  = x_{t + 1} , \ldots ,X^{\left( T \right)}  = x_T } \right|S_H^{\left( t \right)}  = i ,{\bf \Theta }^{\left( {n - 1} \right)} } \right)\, . \nonumber
\end{align}Solving the problem (\ref{EM_formula}) yields an iterative procedure for the estimation of the parameters:\begin{align}
\label{Estimation_transition_prob}
 a_{ij}^{\left( n \right)}  &= \frac{{\sum\nolimits_{t = 1}^{T-1} {P\left( {{\bf x},S_H^{\left( {t } \right)}  = i ,S_H^{\left( t+1 \right)}  = j \left| {{\bf \Theta }^{\left( {n - 1} \right)} } \right.} \right)} }}{{\sum\nolimits_{t = 1}^{T-1} {P\left( {{\bf x},S_H^{\left( {t } \right)}  = i \left| {{\bf \Theta }^{\left( {n - 1} \right)} } \right.} \right)} }} \nonumber\\
   &= \frac{{\sum\nolimits_{t = 1}^{T-1} {\xi _{ij}^{\left( t \right)} } }}{{\sum\nolimits_{j' = 0}^{N_H-1} {\sum\nolimits_{t = 1}^{T-1} {\xi _{i{j'}}^{\left( t \right)} } } }}\, ;
\end{align}\begin{align}
\label{Estimation_state_mean}
 \mu _i^{\left( n \right)}  &= \frac{{\sum\nolimits_{t = 1}^T {x_t P\left( {S_H^{\left( t \right)}  = i \left| {{\bf x},{\bf \Theta }^{\left( {n - 1} \right)} } \right.} \right)} }}{{\sum\nolimits_{t = 1}^T {P\left( {S_H^{\left( t \right)}  = i \left| {{\bf x},{\bf \Theta }^{\left( {n - 1} \right)} } \right.} \right)} }}
  = \frac{{\sum\nolimits_{t = 1}^T {x_t \gamma _i^{\left( t \right)} } }}{{\sum\nolimits_{t = 1}^T {\gamma _i^{\left( t \right)} } }}\, ;
\end{align}\begin{align}
\label{Estimation_state_variance}
  \rho _i^{\left( n \right)}  &= \frac{{\sum\nolimits_{t = 1}^T {\left( {x_t  - \mu _i^{\left( {n - 1} \right)} } \right)^2 P\left( {S_H^{\left( t \right)}  = i \left| {{\bf x},{\bf \Theta }^{\left( {n - 1} \right)} } \right.} \right)} }}{{\sum\nolimits_{t = 1}^T {P\left( {S_H^{\left( t \right)}  = i \left| {{\bf x},{\bf \Theta }^{\left( {n - 1} \right)} } \right.} \right)} }} \nonumber\\
   &= \frac{{\sum\nolimits_{t = 1}^T {\left( {x_t  - \mu _i^{\left( {n - 1} \right)} } \right)^2 \gamma _i^{\left( t \right)} } }}{{\sum\nolimits_{t = 1}^T {\gamma _i^{\left( t \right)} } }}\, ;
\end{align}\begin{align}
\label{Estimation_initial_prob_t0}
 \pi _i^{\left( n \right)}  = \frac{{P\left( {{\bf x},S_H^{\left( 1 \right)}  = i\left| {{\bf \Theta }^{\left( {n - 1} \right)} } \right.} \right)}}{{P\left( {{\bf x}\left| {{\bf \Theta }^{\left( {n - 1} \right)} } \right.} \right)}} = \gamma _i^{\left( 1 \right)}\, ,
\end{align}where $\pi _i^{\left( n \right)} $ represents the posterior probability of the $i^{th}$ solar state conditional on the observation $\bf x$ and the estimated parameter ${\bf \Theta }^{\left( {n - 1} \right)}$. Note that the probability terms $\alpha_i^{(t)}$ and $\beta_i^{(t)}$ in (\ref{Alpha_term}) and (\ref{Beta_term}) can be efficiently carried out using the well-known forward and backward procedures as follows:\begin{align}
\label{update_alpha}
  \alpha _j^{\left( {t + 1} \right)}  = \left( {\sum\nolimits_{i = 0}^{N_H-1} {\alpha _i^{\left( t \right)} } a_{ij}^{\left( {n - 1} \right)} } \right)f_j^{(n-1)} \left( {x_{t + 1} } \right)\, ;
\end{align}\begin{align}
\label{update_beta}
  \beta _i^{\left( t \right)}  =  {\sum\nolimits_{j = 0}^{N_H-1} {a_{ij}^{\left( {n - 1} \right)} f_j^{(n-1)} \left( {x_{t + 1} } \right)} } \beta _j^{\left( {t + 1} \right)} \, ,
\end{align}where the initial values of $\alpha _i^{\left( t \right)}$ and $\beta _i^{\left( t \right)}$ are set as $\alpha _i^{\left( 1 \right)}  = \pi _i^{\left( {n - 1} \right)} f_i \left( {x_1 } \right)$ and $\beta _i^{\left( T \right)}  = 1$. The training procedures from (\ref{Estimation_transition_prob}) to (\ref{Estimation_initial_prob_t0}) are then repeated for several iterations until the values of the parameters get converged. Then, the stationary probability of the Markov chain, $\mbox{\boldmath{$\upsilon $}} = \left[ {P\left( {S_H= 0 } \right), \ldots ,P\left( {S_H= N_H-1 } \right)} \right]^T$, is finally computed by solving the balance equation: \begin{align}
\label{steady_state_prob}
 \left[ \begin{array}{c}
 {\bf A}^{\left( n \right)}  - {\bf I}_{N_H}  \\
 {\bf 1}_{N_H}^T  \\
 \end{array} \right] \mbox{\boldmath{$\upsilon $}} =  \left[ \begin{array}{l}
 {\bf 0}_{N_H}  \\
 1 \\
 \end{array} \right]\, ,
\end{align}where the transition probability matrix is defined as $\left[ {{\bf A}^{\left( n \right)} } \right]_{j,i}  = a_{ij}^{\left( n \right)} $, for $i, j= 0,\ldots,N_H-1$, and $\left[ {\bf A} \right]_{m,n} $ denotes the $(m,n)^{th}$ entry of the matrix $\bf A$.

The training results with respect to the example above are shown in Table \ref{tab:table_em_training} and Fig. \ref{Histogram_irradiance}, where the irradiance measurements are performed every five minutes or fifteen minutes from seven o'clock to seventeen o'clock. We can observe that the histograms of the training results and the measurement data in Fig. \ref{Histogram_irradiance} behave quite similarly when the sampling period is five minutes. Also in Table \ref{tab:table_em_training}, the transition probabilities from the current solar state to the other adjacent states are very small when the measurements are taken at five-minute intervals. In fact, the solar state transition probability largely depends on the sampling period of the measurements, and only a slight increase in the transition probability is observed as the sampling period is increased from five minutes to fifteen minutes.

\begin{table*}
\caption{Training results of the hidden Markov solar power harvesting model.}
\centering
\subtable[Mean, variance and steady state probability.]{
       \begin{tabular}{|c|c|c|c|c|c|c|c|c|}
\hline
 Sampling period & \multicolumn{4}{c|}{5 minutes} & \multicolumn{4}{c|}{15 minutes}\tabularnewline
\hline
\hline
State ($S_{H}=i$) & 0 & 1 & 2 & 3 & 0 & 1 & 2 & 3 \tabularnewline
\hline
$\mu_{i}$ ($10^{4}$ $\mu$W/cm$^2$) & 1.75 & 4.21 & 7.02 & 9.38 & 1.79 & 4.56 & 7.60 & 9.46 \tabularnewline
\hline
 $\rho_i$ & 0.65 & 1.04 & 2.34 & 0.54 & 0.71 & 1.48 & 1.55 & 0.31 \tabularnewline
\hline
 $P$($S_H= i$) & 0.16 & 0.36 & 0.21 & 0.27 & 0.16 & 0.39 & 0.27 & 0.18 \tabularnewline
\hline
\end{tabular}
       \label{tab:firsttable_em_training}
}

\subtable[State transition probability.]{
       \begin{tabular}{|c|c|c|c|c|c|c|c|c|}
\hline
 Sampling period & \multicolumn{4}{c|}{5 minutes} & \multicolumn{4}{c|}{15 minutes}\tabularnewline
\hline
\hline
$a_{ij}$ & $j=0$ & $j=1$ & $j=2$ & $j=3$ & $j=0$ & $j=1$ & $j=2$ & $j=3$\tabularnewline
\hline
$i=0$ & $0.979$ & $0.015$ & $0.006$ & $0$ & 0.938 & 0.057 & 0.005 & 0\tabularnewline
\hline
$i=1$ & $0.005$ & $0.988$ & $0.007$ & $0$ & 0.023 & 0.955 & 0.022 & 0 \tabularnewline
\hline
$i=2$ & $0.006$ & $0.009$ & $0.975$ & $0.010$ & 0 & 0.032 & 0.950 & 0.018 \tabularnewline
\hline
$i=3$ & $0$ & 0 & $0.007$ & $0.993$ & 0.004 & 0 & 0.023 & 0.973 \tabularnewline
\hline
\end{tabular}
       \label{tab:secondtable_em_training}
}
\label{tab:table_em_training}
\end{table*}

The solar power harvesting model is a continuous-time model with respect to the harvested power. In practice, the solar energy is stored in the rechargeable battery to supply the forthcoming communications. The transmission strategy is usually designed on the basis of the required numbers of energy quanta and remains unchanged over a management period of several data packets $T_L$. Here, we map the solar power harvesting model into a discrete energy harvesting model, in which the Markov chain states are described by the numbers of harvested energy quanta. Let $P_U$ be the basic transmission power level of sensor nodes, corresponding to one unit of the energy quantum $E_U= P_U T_L$ during the management period. In addition, for the harvested solar power $P_H$, the obtained energy over the time duration $T_L$ is given by $E_H= P_H T_L$. The numbers of harvested energy quanta, $Q$, at $t=n T_L$ are given as\begin{align}
\label{Accumulated_energy}
  &E_C^{\left( n \right)}  = E_R^{\left( {n - 1} \right)}  + E_H  \, ; \\
\label{Residual_energy_quanta}
  &E_R^{\left( n \right)}  = E_C^{\left( n \right)}  - QE_U ,{\rm   } \,\,Q = \left\lfloor {\frac{{E_C^{\left( n \right)} }}{{E_U }}} \right\rfloor   \, ,
\end{align}where $E_C^{\left( n \right)}$ and $E_R^{\left( n \right)} $ are the accumulated and the residual energy in the capacitor at $t=nT_L$, and $\lfloor \cdot \rfloor$ is the floor function. By assuming that the fluctuation of the harvested power level is quasi-static over many power management runs, it can be analyzed that if $qE_U  \le E_H  \le \left( {q + 1} \right)E_U  $ for some $q$, then the probability of the number of energy quanta, $Q$, can be computed as
\begin{align}
\label{Prob_no_quanta}
  P\left( {Q = i} \right) = \left\{ \begin{array}{l}
 \frac{{E_H  - qE_U }}{{E_U }},{\rm  }\;\;i = q+1 \;\;;\\
 1 - \frac{{E_H  - qE_U }}{{E_U }},{\rm  }\;\;i = q  \;\;;\\
 0,\;\;{\rm otherwise} \;\;.\\
 \end{array} \right.
\end{align}When a sensor node is operated at the $j^{th}$ solar state with the normal distribution $\mathcal{N}\left(x; {\mu _j ,\rho _j } \right)$, the obtained energy $E_H$ is again a normally distributed random variable, which is equal to the solar power per unit area $x$ multiplied by the solar panel area $\Omega_S$, the time duration $T_L$ and the energy conversion efficiency $\vartheta$, i.e., $E_H= x \Omega_S T_L \vartheta$. The conversion efficiency of an energy harvesting device typically ranges between $15\%$ and $20\%$ \cite{S.Sudevalayam11}. Thus, the mean and variance of $E_H$ are respectively given as $\bar \mu _j  = \mu _j \Omega_S T_L \vartheta$ and $\bar \rho _j  = \rho _j \Omega_S^2 T_L^2 \vartheta^2$, and the probability of the number of energy quanta is calculated by using (\ref{Prob_no_quanta}), as follows:
\begin{align}
\label{Prob_no_quanta_jth_state}
   &P\left( {\left. {Q = i} \right|S_H = j} \right)  \\
   &=\left\{ \begin{array}{l}
 \int_{iE_U }^{\left( {i + 1} \right)E_U } {\frac{{\left( {i + 1} \right)E_U  - E_H }}{{E_U }}} \mathcal{N}\left( {E_H ;\bar \mu _j ,\bar \rho _j } \right)dE_H , \;\; i = 0 \;\;;\\
 \int_{iE_U }^{\left( {i + 1} \right)E_U } {\frac{{\left( {i + 1} \right)E_U  - E_H }}{{E_U }}} \mathcal{N}\left( {E_H ;\bar \mu _j ,\bar \rho _j } \right)dE_H  \\
   + \int_{\left( {i - 1} \right)E_U }^{iE_U } {\frac{{E_H  - \left( {i - 1} \right)E_U }}{{E_U }}}\mathcal{ N}\left( {E_H ;\bar \mu _j ,\bar \rho _j } \right)dE_H , \;\; i \ne 0 \;\;.
 \end{array} \right. \nonumber
\end{align}Denote the complementary error function as erfc$(\cdot)$. After some manipulations, we get
\begin{align}
\label{Prob_no_quanta_jth_state_explict}
    &P\left( {\left. {Q = i} \right|S_H = j} \right) \\
    &= \left\{ \begin{array}{l}
 \left( {\left( {i + 1} \right) - \frac{{\bar \mu _j }}{{E_U }}} \right)g_1 \left( {i,\bar \mu _j ,\bar \rho _j } \right) - g_2 \left( {i + 1,\bar \mu _j ,\bar \rho _j } \right), \;\; \\
 \;\;\;\;\;\;\;\;\;\;\;\;\;\;\;\;\;\;\;\;\;\;\;\;\;\;\;\;\;\;\;\;\;\;\;\;\;\;\;\;\;\;\;\;\;\;\;\;\;\;\;\;\;\;\;\;\;\;\;\;\;\;\;\;\;\;\;\;\; i = 0 \;\;; \\
 \left( {\left( {i + 1} \right) - \frac{{\bar \mu _j }}{{E_U }}} \right)g_1 \left( {i,\bar \mu _j ,\bar \rho _j } \right) - g_2 \left( {i + 1,\bar \mu _j ,\bar \rho _j } \right) \\
   + \left( {\frac{{\bar \mu _j }}{{E_U }} - \left( {i - 1} \right)} \right)g_1 \left( {i - 1,\bar \mu _j ,\bar \rho _j } \right) + g_2 \left( {i,\bar \mu _j ,\bar \rho _j } \right),\;\; \\
  \;\;\;\;\;\;\;\;\;\;\;\;\;\;\;\;\;\;\;\;\;\;\;\;\;\;\;\;\;\;\;\;\;\;\;\;\;\;\;\;\;\;\;\;\;\;\;\;\;\;\;\;\;\;\;\;\;\;\;\;\;\;\;\;\;\;\;\;\; i \ne 0 \;\;,
 \end{array} \right. \nonumber
\end{align}where the relevant terms are defined as
\begin{align}
\label{relevant_terms_energy_quanta_prob}
   g_1 \left( {i,\bar \mu _j ,\bar \rho _j } \right) &= \frac{1}{2}\left( {{\rm erfc}\left( {\frac{1}{{\sqrt {2\bar \rho _j } }}\left( {iE_U  - \bar \mu _j } \right)} \right)} \right. \nonumber\\
  &\;\;\;\; \left. {- {\rm erfc}\left( {\frac{1}{{\sqrt {2\bar \rho _j } }}\left( {\left( {i + 1} \right)E_U  - \bar \mu _j } \right)} \right)} \right) \,; \\
   g_2 \left( {i,\bar \mu _j ,\bar \rho _j } \right) &= \sqrt {\frac{{\bar \rho _j }}{{2\pi E_U^2 }}} \left( { \exp \left( { - \frac{1}{{2\bar \rho _j }}\left( {\left( {i - 1} \right)E_U  - \bar \mu _j } \right)^2 } \right)} \right. \nonumber\\
  & \;\;\;\;\left. {- \exp \left( { - \frac{1}{{2\bar \rho _j }}\left( {iE_U  - \bar \mu _j } \right)^2 } \right)} \right)   \, .
\end{align}

\section{Markov Decision Process Using Stochastic Energy Harvesting Models}

We study the adaptive transmissions for sensor communications concerning the channel and battery status, the transmission power, the modulation types, and the stochastic energy harvesting model. Consider a point-to-point communication link with two sensor nodes, where a source node intends to convey data packets to its sink node. Each data packet consists of $L_S$ data symbols at a rate of $R_S$ (symbols/sec), and hence, the packet duration is given by $T_P= L_S/R_S$.

The design framework is formulated as an MDP with the goal of maximizing the long-term net bit rate. As illustrated in Fig. \ref{MDP_state}, the MDP is mainly composed of the state space, the action set, and the state transition probabilities, and it is operated on the time scale of the policy management period $T_L$, covering the time duration of $D$ data packets, i.e., $T_L= DT_P$. Let $\mathcal{S}$ be the state space which is a composite space of the solar state $\mathcal{H}= \{0, \ldots, N_H-1\}$, the channel state $\mathcal{C}= \{0, \ldots, N_C-1\}$ and the battery state $\mathcal{B}= \{0, \ldots, {N_B-1}\}$, i.e., $\mathcal{S}= \mathcal{H} \times \mathcal{C} \times \mathcal{B}$, where $\times$ denotes the Cartesian product. At the $n^{th}$ battery state, we further denote the action space as $\mathcal{A}$ which consists of two-tuple action spaces: transmission power $\mathcal{W}= \{0, \ldots, \min \left\{ {n,N_P  - 1} \right\}\}$ and modulation types $\mathcal{M}= \{0, \ldots, N_M-1\}$. Since the transition probabilities of the channel and battery states are independent of each other, the transition probability from the state ${\left( {S_H, S_C ,S_B } \right) = \left( {j,i,n} \right)}$ to the state ${\left( {S_H, S_C ,S_B } \right) = \left( {j',i',n'} \right)}$ with respect to the action $\left( {W,M} \right) = \left( {w,m} \right)$ under the $j^{th}$ solar state can be formulated as
\begin{align}
\label{Overall_steady_state_probability}
 &P_{w,m} \left( {\left. {\left( {S_H ,S_C ,S_B } \right) = \left( {j',i',n'} \right)} \right|\left( {S_H ,S_C ,S_B } \right) = \left( {j,i,n} \right)} \right) \nonumber\\
 & = P\left( {\left. {S_H  = j'} \right|S_H  = j} \right)P\left( {\left. {S_C  = i'} \right|S_C  = i} \right) \nonumber\\
 &\;\;\;\;\;\;\;\;\;\;\;\;\;\;\;\;\;\;\;\;\;\;\;\cdot P_w \left( {\left. {S_B  = n'} \right|\left( {S_H ,S_B } \right) = \left( {j,n} \right)} \right) \,,
\end{align}where the transition of the battery states is irrespective of the adopted modulation type. Note that the transition probability of the solar states, $P\left( {\left. {S_H  = j'} \right|S_H  = j} \right)$, can be directly obtained by using the training results in (\ref{Estimation_transition_prob}). In the following, we elaborate on each of the components of the MDP in Fig. \ref{MDP_state} before describing the solution of the Bellman optimality equation.

\begin{figure}[t]
\centering
\includegraphics[width=0.4\textwidth]{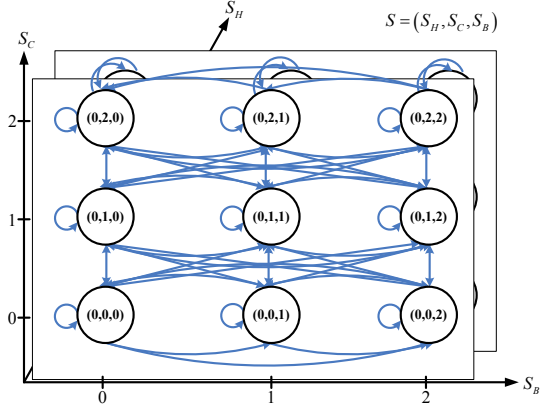}
\caption{Markov chain for the Markov decision process ($N_H=2$, $N_C=3$ and $N_B=3$).}
\label{MDP_state}
\end{figure}

\subsection{Actions of Transmission Power and Modulation Types}
When the action $(w, m) \in \mathcal{W} \times \mathcal{M}$ is chosen by the sensor node, the transmission power and modulation levels are respectively set as $w P_U$ and as $2^{\chi_m}$-ary phase shift keying (PSK) or quadrature amplitude modulation (QAM), e.g., QPSK, 8PSK and 16QAM, during the policy management period, where $\chi_m$ represents the number of information bits in each data symbol. Remember that $P_U$ is the basic transmission power level of the sensor node if data transmission takes place. On the other hand, if $w=0$, the node remains silent without transmitting data packets.

\subsection{Channel State and State Transition Probability}
The wireless channel is quantized using a finite number of thresholds ${\bf \Gamma } = \left\{ {0 = \Gamma _0 ,\Gamma _1 , \ldots ,\Gamma _{N_C }   }\right.$ $\left. = \infty\right\}$, where $\Gamma _i  < \Gamma _j $ for all $i < j$. The Rayleigh fading channel is said to be in the $i^{th}$ channel state, for $i= 0, \ldots, N_C-1$, if the instantaneous channel power, $\gamma$, belongs to the interval $\left[ {\Gamma _i ,\Gamma _{i + 1} } \right)$. We assume that the wireless channel fluctuates slowly and the policy management period is shorter than the coherence time of the fading channel. Hence, the channel gain remains quasi-static during each transmission period of $T_L$, and the channel state transition occurs only from the current state to its neighboring states. The stationary probability of the $i^{th}$ state is
\begin{align}
\label{Channel_steady_state_probability}
  P\left( {S_C  = i} \right) &= \int_{\Gamma _i }^{\Gamma _{i + 1} } {\frac{1}{{\gamma _0 }}\exp \left( { - \frac{\gamma }{{\gamma _0 }}} \right)} d\gamma  \\
  &= \exp \left( { - \frac{{\Gamma _i }}{{\gamma _0 }}} \right) - \exp \left( { - \frac{{\Gamma _{i + 1} }}{{\gamma _0 }}} \right) \, , \nonumber
\end{align}where $\gamma_0= \mathbb{E}\left[ \gamma \right]$ is the average channel power. Define $h\left( \gamma  \right) = \sqrt {{{2\pi \gamma } \mathord{\left/
 {\vphantom {{2\pi \gamma } {\gamma _0 }}} \right.
 \kern-\nulldelimiterspace} {\gamma _0 }}} f_D \exp \left( {{{ - \gamma } \mathord{\left/
 {\vphantom {{ - \gamma } {\gamma _0 }}} \right.
 \kern-\nulldelimiterspace} {\gamma _0 }}} \right)$, where $f_D$ is the maximum Doppler frequency, normalized by the policy management rate $1/T_L$. The state transition probabilities are determined by \cite{H.S.Wang95} \begin{align}
\label{Channel_state_transition_probability}
  &P\left( {\left. {S_C  = k} \right|S_C  = i} \right) \\
  &= \left\{ \begin{array}{l}
 \frac{{h\left( {\Gamma _{i + 1} } \right)}}{{P\left( {S_C  = i} \right)}},\;\;k = i + 1,\;\;i = 0, \ldots ,N_C  - 2 \,; \\
 \frac{{h\left( {\Gamma _i } \right)}}{{P\left( {S_C  = i} \right)}},\;\;k = i - 1,\;\;i = 1, \ldots ,N_C  - 1 \,;\\
 1 - \frac{{h\left( {\Gamma _i } \right)}}{{P\left( {S_C  = i} \right)}} - \frac{{h\left( {\Gamma _{i + 1} } \right)}}{{P\left( {S_C  = i} \right)}},\;\;k = i,\;\; i = 1, \ldots ,N_C  - 2 \,,\nonumber
 \end{array} \right.
\end{align}and the transition probabilities of $P\left( {\left. {S_C  = i} \right|S_C  = i} \right)$ for the boundaries are given by\begin{align}
\label{Channel_state_transition_probability_2}
 &P\left( {\left. {S_C  = 0} \right|S_C  = 0} \right) = 1 - P\left( {\left. {S_C  = 1} \right|S_C  = 0} \right) \,; \nonumber\\
 &P\left( {\left. {S_C  = N_C  - 1} \right|S_C  = N_C  - 1} \right) \nonumber\\
 &\;\;\;\;\;\;\;\;\;\;\;= 1 - P\left( {\left. {S_C  = N_C  - 2} \right|S_C  = N_C  - 1} \right) \,.
\end{align}

\subsection{Battery State and State Transition Probability}
When the sensor node is run at the $n^{th}$ battery state, it means that the available energy in the battery is stored up to $n$ energy quanta, i.e., $n E_U$. Remember that one energy quantum, $E_U$, represents the basic energy unit required for the sensor node to adopt the basic transmission power level $P_U$ over the time duration $T_L$. At the $n^{th}$ battery state, the possible action that can be performed by the sensor is from $0$ to $\min \left\{ {n,N_P  - 1} \right\}$, and the $w^{th}$ power action will consume a total of $w$ energy quanta for data transmission. In particular, the sensor is unable to make any transmission when the energy is completely depleted at the $0^{th}$ state. Once the underlying parameters of the $N_H$ solar states are appropriately estimated through the measurement data, the state transition probabilities for the $n^{th}$ battery state and the $w^{th}$ power action under the $j^{th}$ solar state can be constructed by exploiting (\ref{Prob_no_quanta_jth_state_explict}), as follows:\begin{align}
\label{Battery_state_transition_probability}
 &P_w \left( {\left. {S_B  = k} \right|\left( {S_H,S_B } \right) = \left( {j,n} \right)} \right) \\
 & = \left\{ \begin{array}{l}
 P\left( {\left. {Q = k - n + w} \right|S_H = j} \right),\;\; \\
 \;\;\;\;\;\;\;\;\;\;\;\;\;\;\;\;\;\;\;\;\;\;\;\;\;\;\;\;\;\;\;\;\;\;\;\;\;\;\;\;\;\; k = n - w, \ldots ,N_B  - 2 \,; \\
 1 - \sum\nolimits_{i = 0}^{N_B  - 2 - n + w} {P \left( {\left. {Q  = i} \right| {S_H } =  {j}} \right)} ,\;\; k = N_B  - 1 \,,
 \end{array} \right. \nonumber
\end{align}for $n=0,\ldots,N_B-1$ and $w=0,\ldots,\min \left\{ {n,N_P-1 } \right\}$.

\subsection{Reward Function}
We adopt the average number of good bits per packet transmission as our reward function. It is assumed that the sink node periodically feeds back the instantaneous channel state information to the source node for planning the next transmission policy. Let $P_{e,b} \left( {\left( {S_C ,S_B, W,M} \right) = \left( {i,n, w,m} \right)} \right)$ be the average bit error rate (BER) at the $i^{th}$ channel state and the $n^{th}$ battery state when the action $\left( {W,M} \right) = \left( {w,m} \right)$ is taken, for $w=0,\ldots,\min \left\{ {n,N_P-1 } \right\}$ and $m=0,\ldots,N_M-1$. By applying the upper bound of the Q-function $Q\left( x \right) \le \frac{1}{2}\exp \left( { - \frac{{x^2 }}{2}} \right)$, it can be computed as\begin{align}
\label{Reward_function_BER}
&P_{e,b} \left( {\left( {S_C ,S_B,W,M} \right) = \left( {i,n,w,m} \right)} \right) \\
&= \frac{  \int_{\Gamma _i }^{\Gamma _{i + 1} } {\alpha _m Q\left( {\sqrt {\frac{{\beta _m wP_U \gamma }}{{N_0 }}} } \right)} \frac{1}{{\gamma _0 }}\exp \left( { - \frac{\gamma }{{\gamma _0 }}} \right)d\gamma   }{{\int_{\Gamma _i }^{\Gamma _{i + 1} } {\frac{1}{{\gamma _0 }}\exp \left( { - \frac{\gamma }{{\gamma _0 }}} \right)d\gamma } }} \nonumber\\
 & \le  \frac{  \frac{{\alpha _m }}{{w\beta _m \gamma _U  + 2}}   }{\exp \left( { - \frac{{\Gamma _i }}{{\gamma _0 }}} \right) - \exp \left( { - \frac{{\Gamma _{i + 1} }}{{\gamma _0 }}} \right)} \nonumber\\
 &\;\;\;\cdot \Bigg( \exp \left( { - \frac{1}{{2\gamma _0 }}\left( {w\beta _m \gamma _U  + 2} \right)\Gamma _i } \right) \nonumber\\
 &\;\;\;\;\;\;\;\;\;\;\;\;\;\;\;\;\;\;\;\;\;\;\;\;\;\;\;- \exp \left( { - \frac{1}{{2\gamma _0 }}\left( {w\beta _m \gamma _U  + 2} \right)\Gamma _{i + 1} } \right) \Bigg) \nonumber\\
 & \triangleq \eta \left( {i,n,w,m} \right)\,, \nonumber
\end{align}where $\alpha_m$ and $\beta_m$ are modulation specific constants, $N_0$ is the noise power, and $\gamma _U  = {{P_U \gamma _0 } \mathord{\left/
 {\vphantom {{P_U \gamma _0 } {N_0 }}} \right.
 \kern-\nulldelimiterspace} {N_0 }}$ is the average signal-to-noise power ratio (SNR) when the basic transmission power level is adopted. Hence, the probability of successful packet transmission (i.e., all $\chi_mL_S$ bits in a packet are successfully detected) is expressed as
 \begin{align}
\label{Reward_function_PER}
&P_{f,k} \left( {\left( {S_C ,S_B,W,M} \right) = \left( {i,n,w,m} \right)} \right) \\
& = \left( {1 - P_{e,b} \left( {i,n,w,m} \right)} \right)^{\chi_mL_S }  \,.\nonumber
\end{align}If the sensor fails to decode the received data packet, the retransmission mechanism is employed in the sensor communications. Let $Z$ be the total number of retransmissions required to successfully convey a data packet. By assuming that each transmission is independent, the variable $Z$ can be expressed as a geometric random variable, and the average number of retransmissions for the successful reception of a packet is given by
\begin{align}
\label{Reward_function_retransmission_number}
\mathbb{E}\left[ Z  \right] = {1 \mathord{\left/
 {\vphantom {1 {P_{e,k} \left( {i,n,w,m} \right)}}} \right.
 \kern-\nulldelimiterspace} {P_{f,k} \left( {i,n,w,m} \right)}}   \,.
\end{align}
Since $T_L= DT_P$, the number of effective data packets due to retransmission during each management period is in average given as
\begin{align}
\label{Average_number_successful_packet}
D_{E}  = \frac{D}{{\mathbb{E}\left[ Z \right]}} = \frac{{T_L }}{{\mathbb{E}\left[ Z \right]T_P }}   \,.
\end{align}
From (\ref{Reward_function_BER})-(\ref{Average_number_successful_packet}), the net bit rate can therefore be lower bounded by
\begin{align}
\label{Reward_function_lower_bound}
 &G_{w,m}\left( {\left( {S_C ,S_B } \right) = \left( {i,n} \right)} \right)= \frac{1}{{T_L }}D_E \chi _m L_S \\
 & = \frac{1}{{T_P }} \chi_mL_S \left( {1 - P_{e,b} \left( {i,n,w,m} \right)} \right)^{\chi_mL_S }  \nonumber\\
 &\ge \frac{1}{{T_P }} \chi_mL_S \left( {1 - \eta \left( {i,n,w,m} \right)} \right)^{\chi_mL_S } \,.\nonumber
\end{align}

\begin{definition}
\label{definition1}
The reward function for the action $(W, M)= (w, m)$ at the state $(S_C, S_B)= (i, n)$ is defined as
\begin{align}
\label{Reward_function_good_bits}
&R_{w,m}\left( {\left( {S_C ,S_B } \right) = \left( {i,n} \right)} \right) \\
&= \left\{ \begin{array}{l}
 0, \;\; w=0 \,;\\
 \frac{1}{{T_P }} \chi_mL_S \left( {1 - \eta \left( {i,n,w,m} \right)} \right)^{\chi_mL_S } , \;\; w \in \mathcal{W} \backslash \left\{ 0 \right\} \,.
 \end{array} \right.\nonumber
\end{align}
The reward function has the following properties:
\begin{enumerate}[(a)]
\item $R_{w,m}\left( {\left( {S_C ,S_B } \right) = \left( {i,n} \right)} \right)=0$ for $w=0$, because no data transmission occurs when the transmission power is zero.
\item $R_{w,m}\left( {\left( {S_C ,S_B } \right) = \left( {i,n} \right)} \right) = R_{w',m}\left( {\left( {S_C ,S_B } \right) = \left( {i,n'} \right)} \right)$ for any $w=w'$, because the immediate reward is independent of the battery state.
\item $R_{w,m}\left( {\left( {S_C ,S_B } \right) = \left( {i,n} \right)} \right) \geq R_{w,m}\left( {\left( {S_C ,S_B } \right) = \left( {i',n} \right)} \right)$ for any $i \geq i'$, which means a higher immediate reward is obtained as the channel condition improves.
\end{enumerate}
\end{definition}

\subsection{Transmission Policies}
Two transmission policies are implemented regarding the affordable actions in the action set $\mathcal{A}= \mathcal{W} \times \mathcal{M} $.
\begin{definition}
\label{definition2}
(Composite policy) A transmission policy is composite, if $N_P  \ge N_B $. The action set at the $n^{th}$ battery state is given by $\mathcal{A}= \left\{ {0, \ldots ,n} \right\} \times \left\{ {0, \ldots ,N_M  - 1} \right\}$.
\end{definition}

\begin{definition}
\label{definition3}
(On-off policy) A transmission policy is on-off, if $N_P  = 2$ and $N_M  = 1$. The action set at the $n^{th}$ battery state is given by $\mathcal{A} = \left\{ {0, \ldots ,\min \left\{ {n,1} \right\}} \right\} \times \left\{ 0 \right\}$.
\end{definition}

In the composite policy, the power action could be unconditional as long as the resultant energy consumption during the policy management period is below the battery supply. On the contrary, only a single power and modulation level is accessible in the on-off policy whenever the sensor is active for data transmission. The composite policy undoubtedly has better performance than the on-off policy, whereas the later one, as its name suggests, operates in a relatively simple on-off switching mode for data transmission.

\section{Optimization of Transmission Policies}
The main goal of the MDP is to find a decision policy $\pi(s): \mathcal{S}\rightarrow \mathcal{A}$ that specifies the optimal action in the state $s$ and maximizes the objective function. Since we are interested in maximizing some cumulative functions of the random rewards in the Markov chain, the expected discounted infinite-horizon reward is formulated by using (\ref{Reward_function_good_bits}):
\begin{align}
\label{Long_term_expected_reward_function}
&V_\pi  \left( {s_0 } \right) = \mathbb{E}_\pi  \left[ {\sum\nolimits_{k = 0}^\infty  {\lambda ^k R_{\pi \left( {s_k } \right)}\left( {s_k} \right)} } \right],\;\;\nonumber \\
&\;\;\;\;\;\;\;\;\;\;\;\;\;\;\;\;\;\;\;\;\;\;\;\;\;\;\;\;\;\;\;\;\;\;\;\;\;\;\;\;\;s_k  \in \mathcal{S},\;\; \pi \left( {s_k } \right) \in \mathcal{A}  \,,
\end{align}where $V_\pi  \left( {s_0 } \right)$ is the expected reward starting from the initial state $s_0$ and continuing with the policy $\pi$ from then on, and $0\leq \lambda \leq 1$ is a discount factor. The adjustment of the discount factor provides a wide range of performance characteristics. Particularly, the long run average objective can be closely approximated by using a discount factor close to one. It is known that the optimal value of the expected reward is unrelated to the initial state if the states of the Markov chain are assumed to be recurrent. From (\ref{Overall_steady_state_probability}) and (\ref{Long_term_expected_reward_function}), there exists an optimal stationary policy $\pi^* \left( s \right)$ that satisfies the Bellman's equation:
\begin{align}
\label{Bellman_equation}
&V_{\pi ^ *  } \left( s \right) = \mathop {\max }\limits_{a \in \mathcal{A}} \left( {R_a\left( {s} \right) + \lambda \sum\limits_{s'\in \mathcal{S}} {P_a\left( {\left. {s'} \right|s} \right)V_{\pi ^ *  } \left( {s'} \right)} } \right),\;\; \nonumber \\
& \;\;\;\;\;\;\;\;\;\;\;\;\;\;\;\;\;\;\;\;\;\;\;\;\;\;\;\;\;\;\;\;\;\;\;\;\;\;\;\;\;\;\;\;\;\;\;\;\;\;\;\;\;\;\;\;\;\;\;\;\;\;\;\;\; s \in \mathcal{S}  \,.
\end{align}The well-known value iteration approach, which is known as dynamic programming, can be applied to iteratively find the optimal policy:
\begin{align}
\label{Value_iteration_1}
&V_{i + 1}^a \left( s \right) = R_a\left( {s} \right) + \lambda \sum\limits_{s' \in \mathcal{S}} {P_a\left( {\left. {s'} \right|s} \right)V_i \left( {s'} \right)} ,\;\; \nonumber \\
&\;\;\;\;\;\;\;\;\;\;\;\;\;\;\;\;\;\;\;\;\;\;\;\;\;\;\;\;\;\;\;\;\;\;\;\;\;\;\;\;\;\;\;\;\;\;\;\;\;\;\; s \in \mathcal{S},\;\;a \in \mathcal{A}  \,;
\end{align}
\begin{align}
\label{Value_iteration}
V_{{i + 1} } \left( s \right) = \mathop {\max }\limits_{a \in \mathcal{A}} \left\{ {V_{i + 1}^a \left( s \right)} \right\},\;\; s \in \mathcal{S}  \,,
\end{align}where $i$ is the iteration index, and the initial value of $V_{0 } \left( s \right)$ is set as zero for all $s\in \mathcal{S}$. The update rule is repeated for several iterations until a stop criterion is satisfied, i.e., $\left| {V_{{i + 1} } \left( s \right) - V_i \left( s \right)} \right| \le \varepsilon $.

In real applications, the channel state of the communication link can be reliably obtained at the transmitter via channel feedback information. The belief of the solar state can be calculated from the observation prior to the action decision for each policy management period. Let ${x}^{\left( t \right)}$ be the average value of the measured solar data during the $t^{th}$ management period, and $\zeta _j^{\left( {t - 1} \right)}  = P\left( {\left. {S_H^{(t-1)}  = j} \right|{x}^{\left( 1 \right)} , \ldots ,{x}^{\left( {t - 1} \right)} } \right)$ be the belief of the $j^{th}$ solar state according to the historical observation up to the $(t-1)^{th}$ period. With the solar power harvesting model, the belief information at the $t^{th}$ period can be updated using Bayes' rule as follows:
\begin{align}
\label{blief_information_update}
\zeta _j^{\left( t \right)}  = \frac{{\sum\nolimits_{i = 0}^{N_H  - 1} {\zeta _i^{\left( {t - 1} \right)} a_{ij} {f_j \left( {x^{\left( t \right)} } \right)} } }}{{\sum\nolimits_{j' = 0}^{N_H  - 1} {\sum\nolimits_{i' = 0}^{N_H  - 1} {\zeta _{i'}^{\left( {t - 1} \right)} a_{i'j'} {f_{j'} \left( {x^{\left( t \right)} } \right)} } } }} \,,
\end{align}where $f_j \left( x \right)$ and $a_{ij}$ are the likelihood function and the state transition probability defined in (\ref{xi_term}) and (\ref{Estimation_transition_prob}), respectively. The final task is to apply the belief information $\zeta _j^{\left( {t} \right)} $ for deciding the action at each management period. We consider the following mixed strategy. Assuming that the current channel and battery states, $\left( {S_C ,S_B } \right) = \left( {i,n} \right)$, are known at each period, the action corresponding to the $j^{th}$ solar state is chosen with probability proportional to the belief information $\zeta _j^{\left( t \right)}$.

\section{Optimal On-Off Transmission Policies}

\subsection{Threshold Structure of Transmission Policies}
To facilitate analysis, we focus on a simple on-off transmission policy and drop the modulation type index $m$ for simplicity, i.e., $a=w \in\{0,1\}$. From (\ref{Prob_no_quanta_jth_state_explict}), (\ref{Overall_steady_state_probability}) and (\ref{Channel_state_transition_probability})-(\ref{Battery_state_transition_probability}), the expected reward function with respect to the action $a$ in (\ref{Value_iteration_1}) can be rewritten as
\begin{align}
\label{simple_expected_reward_reformulation}
 &V_{i + 1}^a \left( {z,x,y} \right) \\
 &= R_a \left( {x,y} \right) + \lambda \sum\limits_{j=0}^{N_H-1}{P\left( {\left. {S_H  = j} \right|S_H  = z} \right)} \nonumber \\
 &\;\;\cdot \sum\limits_{l = \max \left\{ {0,x - 1} \right\}}^{\min \left\{ {x + 1, N_C-1} \right\} } {P\left( {\left. {S_C  = l} \right|S_C  = x} \right)}  \nonumber\\
 &\;\;\cdot \sum\limits_{q = 0}^\infty  {P\left( {\left. {Q = q} \right|S_H  = z} \right)V_i \left( {j, l,\min \left\{ {N_B-1 ,y - a + q} \right\}} \right)}
 \nonumber\\
 &= R_a \left( {x,y} \right) + \lambda \cdot \mathbb{E}_{z,x,y}\left[ {V_i \left( {j, l,\min \left\{ {N_B-1 ,y - a + q} \right\}} \right)} \right] \,, \nonumber
\end{align}where the subscript in $\mathbb{E}_{z,x,y}\left[\cdot\right]$ is used to indicate the associated solar, channel and battery states.

\begin{lemma}
\label{lemma1}
For any fixed solar state $z \in \mathcal{H}$ and channel state $x \in \mathcal{C}$, $V_i^a \left( {z, x,y - 1} \right) \le V_i^a \left( {z, x,y} \right)$, $\forall y \in \mathcal{B} \backslash \{0\}$ and $a=0, 1$. Moreover, $V_i \left( {z, x,y - 1} \right) \le V_i \left( {z, x,y} \right)$, $\forall y \in \mathcal{B}\backslash \{0\}$
\end{lemma}

\begin{proof}
 From (\ref{Value_iteration}), if $V_i^a \left( {z, x,y - 1} \right) \le V_i^a \left( {z, x,y} \right)$ is satisfied, it implies
\begin{align}
\label{Lemma1_proof_relationship}
 &V_i \left( {z, x,y - 1} \right) = \max \limits_{a \in \left\{ {0,1} \right\}} \left\{V_i^a \left( {z, x,y - 1} \right)\right\} \nonumber\\
 &\;\;\;\;\;\;\;\;\;\;\;\;\;\;\;\;\;\;\le \max \limits_{a \in \left\{ {0,1} \right\}} \left\{V_i^a \left( {z,x,y} \right)\right\} = V_i \left( {z,x,y} \right) \,.
\end{align}
We prove the lemma by the induction. From (\ref{simple_expected_reward_reformulation}) and the initial condition $V_0(s)=0$, the statement is held for $i=1$ because  \begin{align}
\label{Lemma1_proof_i1_first}
  &V_1^a \left( {z,x,y - 1} \right) = R_a \left( {x,y - 1} \right) \nonumber\\
  &\;\;\;\;\;\;\;\;\;\;\;\;\;\;\;\;\;\;= R_a \left( {x,y} \right) = V_1^a \left( {z,x,y} \right), \;\; a \in \{0, 1\} \,.
\end{align}Hence, we obtain $V_1 \left( {z,x,y - 1} \right) = V_1 \left( {z,x,y} \right)$. Assume $i=k$ holds, and for any $z\in \mathcal{H}$ and $x\in \mathcal{C}$, it gives $V_{k}\left( {z,x,y - 1} \right) \le V_{k} \left( {z,x,y} \right)$, $\forall y \in \mathcal{B} \backslash \{0\}$. Using (\ref{simple_expected_reward_reformulation}), we prove that for $i=k+1$:
\begin{align}
\label{Lemma1_proof_ik+1_first}
 &V_{k + 1}^a \left( {z,x,y} \right) - V_{k + 1}^a \left( {z,x,y - 1} \right)  \\
 &= \lambda \sum\limits_{j=0}^{N_H-1}{P\left( {\left. {S_H  = j} \right|S_H  = z} \right)} \nonumber\\
 & \cdot \sum\limits_{l = \max \left\{ {0,x - 1} \right\}}^{\min \left\{ {x + 1, N_C-1} \right\} } {P\left( {\left. {S_C  = l} \right|S_C  = x} \right) \sum\limits_{q = 0}^\infty {P\left( {\left. {Q = q} \right|S_H  = z} \right)} }  \nonumber\\
 & \cdot  \Big( {V_k \left( {j, l,\min \left\{ {N_B-1 ,y - a + q} \right\}} \right)} \nonumber \\
 &\;\;\;\;\;\;\;\;\;\;\;\;\;\;\;\;{- V_k \left( {j,l,\min \left\{ {N_B-1 ,y - 1 - a + q} \right\}} \right)} \Big)  \ge 0 \,.\nonumber
\end{align}This thereby implies that $V_{k+1} \left( {z,x,y - 1} \right) \le V_{k+1} \left( {z,x,y} \right)$, and the statement holds for $i=k+1$.

\end{proof}

\begin{theorem} For the optimal transmission policy, the long-term expected reward is monotonically non-decreasing with respect to the battery state. In other words, for any fixed solar state $z \in \mathcal{H}$ and channel state $x \in \mathcal{C}$, $V_{\pi ^ *  } \left( {z,x,y - 1} \right) \le V_{\pi ^ *  } \left( {z,x,y} \right)$, $\forall y \in \mathcal{B}\backslash \{0\}$.
\end{theorem}

\begin{proof} By applying Lemma \ref{lemma1} and following the value iteration algorithm, the theorem is proved when the algorithm has converged.
\end{proof}

Now we turn to describing the structure of the on-off transmission policies. Since no transmission (i.e., $a=0$) is the only action when the battery state is zero, we concentrate on the actions for $y \in \mathcal{B}\backslash \{0\}$ in the following.

\begin{lemma}
\label{lemma2}
For each $z \in \mathcal{H}$, $x \in \mathcal{C}$ and $y \in \mathcal{B}\backslash \{0\}$, define two difference functions:\begin{align}
\label{Lemma2_definition_1}
\Theta _i \left( {z,x,y} \right) = V_i^1 \left( {z,x,y} \right) - V_i^0 \left( {z,x,y} \right) \,;
\end{align}\begin{align}
\label{Lemma2_definition_2}
\Lambda _i \left( {z,x,y} \right) =
 &\mathbb{E}_{z,x,y} \Big[ {V_i^1 \left( {j,l,\min \left\{ {N_B-1 ,y + q} \right\}} \right)} \nonumber\\
 &\;\;-{ V_i^0 \left( {j,l,\min \left\{ {N_B-1 ,y - 1 + q} \right\}} \right)} \Big] \,.
\end{align}
The function $\Theta _i \left( {z,x,y} \right)$ is monotonically non-decreasing in $y \in \mathcal{B}\backslash \{0\}$, if the function $\Lambda _t \left( {z,x,y} \right) $ is non-increasing in $y \in \mathcal{B}\backslash \{0\}$, $\forall t<i$, $z \in \mathcal{H}$ and $x \in \mathcal{C}$.
\end{lemma}

\begin{proof}
We use induction to prove this lemma. When $i=1$, the statement is true because
$\Theta _1 \left( {z,x,y} \right) = V_1^1 \left( {z,x,y} \right) - V_1^0 \left( {z,x,y} \right) = R_1 \left( {x,y} \right)$, for $y \neq 0$, and the reward function $R_1 \left( {x,y} \right)$ keeps the same value in $y \in \mathcal{B}\backslash \{0\}$ for any given $x \in \mathcal{C}$.

Assume $i=k$ holds, the function $\Theta _{k} \left( {z,x,y} \right)$ is non-decreasing in $y \in \mathcal{B}\backslash \{0\}$, $\forall z \in \mathcal{H}$ and $\forall x \in \mathcal{C}$. It immediately implies that the following two functions are both non-decreasing in $y$:
\begin{align}
\label{Lemma2_proof_ik_max}
\Delta _k^{\max } \left( {z,x,y} \right) = \max \left\{ {0,\Theta _k \left( {z,x,y} \right)} \right\} \geq 0 \,;
\end{align}\begin{align}
\label{Lemma2_proof_ik_min}
\Delta _k^{\min } \left( {z,x,y} \right) = \min \left\{ {0,\Theta _k \left( {z,x,y} \right)} \right\} \leq 0 \,.
\end{align}
For $i=k+1$, the difference function $\Theta _{k+1} \left( {z,x,y} \right)$ can be derived from (\ref{Value_iteration}) and (\ref{simple_expected_reward_reformulation}) as follows:\begin{align}
\label{Lemma2_proof_ik+1_first}
 &\Theta _{k + 1} \left( {z,x,y} \right) = V_{k + 1}^1 \left( {z,x,y} \right) - V_{k + 1}^0 \left( {z,x,y} \right) \\
  &= R_1 \left( {x,y} \right) + \lambda \mathbb{E}_{z,x,y}\left[ {V_k \left( {j,l,\min \left( {N_B-1 ,y - 1 + q} \right)} \right)} \right] \nonumber\\
  &\;\;\;\;- R_0 \left( {x,y} \right) - \lambda \mathbb{E}_{z,x,y}\left[ {V_k \left( {j,l,\min \left( {N_B-1 ,y + q} \right)} \right)} \right] \nonumber\\
  &= R_1 \left( {x,y} \right) - R_0 \left( {x,y} \right) \nonumber\\
  &\;\;\;\;+ \lambda \mathbb{E}_{z,x,y}\big[ \max \big\{ V_k^0 \left( {j,l,\min \left\{ {N_B-1 ,y - 1 + q} \right\}} \right), \nonumber\\
  &\;\;\;\;\;\;\;\;\;\;\;\;\;\;\;\;\;\;\;\;\;\;\;\;\;\;\;\;\;\;\;V_k^1 \left( {j,l,\min \left\{ {N_B-1 ,y - 1 + q} \right\}} \right) \big\} \big] \nonumber\\
  &\;\;\;\;- \lambda \mathbb{E}_{z,x,y}\big[ \max \big\{ V_k^0 \left( {j,l,\min \left\{ {N_B-1 ,y + q} \right\}} \right),\nonumber\\
  &\;\;\;\;\;\;\;\;\;\;\;\;\;\;\;\;\;\;\;\;\;\;\;\;\;\;\;\;\;\;\;V_k^1 \left( {j,l,\min \left\{ {N_B-1 ,y + q} \right\}} \right) \big\} \big]\,. \nonumber
\end{align}
Inserting (\ref{Lemma2_proof_ik_max}) and (\ref{Lemma2_proof_ik_min}) into (\ref{Lemma2_proof_ik+1_first}) yields \begin{align}
\label{Lemma2_proof_ik+1_seond}
& \Theta _{k + 1} \left( {z,x,y} \right) \\
&  = R_1 \left( {x,y} \right) \nonumber\\
&  \;\;\; + \lambda \mathbb{E}_{z,x,y}\big[ V_k^0 \left( {j,l,\min \left\{ {N_B-1 ,y - 1 + q} \right\}} \right) \nonumber\\
& \;\;\;\;\;\;\;\;\;\;\;\;\;\;\;\;\;\;\; \;\;\;\;\;\;\; + \Delta _k^{\max } \left( {j,l,\min \left\{ {N_B-1 ,y - 1 + q} \right\}} \right) \big] \nonumber\\
&  \;\;\; - \lambda \mathbb{E}_{z,x,y}\big[ V_k^1 \left( {j,l,\min \left\{ {N_B-1 ,y + q} \right\}} \right)  \nonumber\\
& \;\;\;\;\;\;\;\;\;\;\;\;\;\;\;\;\;\;\; \;\;\;\;\;\;\; - \Delta _k^{\min } \left( {j,l,\min \left\{ {N_B-1 ,y + q} \right\}} \right) \big] \nonumber\\
&  = R_1 \left( {x,y} \right) - \lambda \Lambda _k \left( {z,x,y} \right) \nonumber\\
&  \;\;\;+ \lambda \mathbb{E}_{z,x,y}\big[ \Delta _k^{\max } \left( {j,l,\min \left\{ {N_B-1 ,y - 1 + q} \right\}} \right) \big] \nonumber\\
& \;\;\;+ \lambda \mathbb{E}_{z,x,y}\big[ {\Delta _k^{\min } \left( {j,l,\min \left\{ {N_B-1 ,y + q} \right\}} \right)} \big] \,.\nonumber
\end{align}
According to the non-decreasing property of the functions $\Delta _k^{\max } \left( {z,x,y} \right)$, $\Delta _k^{\min } \left( {z,x,y} \right)$ and $R_1 \left( {x,y} \right)$, it can be shown from (\ref{Lemma2_proof_ik+1_seond}) that $\Theta _{k + 1} \left( {z,x,y} \right)$ preserves the non-decreasing property in $y  \in \mathcal{B}\backslash \{0\}$, if $\Lambda _k \left( {z,x,y} \right)$ is non-increasing in $y \in \mathcal{B}\backslash \{0\}$, $\forall z \in \mathcal{H}$ and $\forall x \in \mathcal{C}$.
\end{proof}

In fact, the validity of the non-decreasing property of $\Theta _i \left( {z,x,y} \right)$ relies on the transition probabilities of the solar states, channel states and battery states, and this property is not necessarily satisfied in $z \in \mathcal{H}$ and $x \in \mathcal{C}$. Below we show that the function $\Lambda _t \left( {z, x,y} \right)$ is indeed non-increasing in the direction along the battery states for a given solar state and channel state, and the following theorem is provided.

\begin{theorem}
\label{theorem2}
For a given solar state $z \in \mathcal{H}$ and channel state $x \in \mathcal{C}$, the difference function $\Theta _i \left( {z,x,y} \right)$ is monotonically non-decreasing in $y \in \mathcal{B}\backslash \{0\}$, and the optimal transmission policy generated by the value iteration algorithm has a threshold structure.
\end{theorem}

\begin{proof}
We first show that $\Lambda _t \left( {z, x,y + 1} \right) - \Lambda _t \left( {z, x,y} \right) \leq 0$, for $y=1,\ldots,N_B-2$, in the following. It can be derived from the definition in Lemma \ref{lemma2} that
\begin{align}
\label{theorem2_proof_1}
 &\Lambda _t \left( {z, x,y + 1} \right) - \Lambda _t \left( {z, x,y} \right)  = \sum\limits_{j=0}^{N_H-1}{P\left( {\left. {S_H  = j} \right|S_H  = z} \right)} \nonumber \\
 & \;\;\;\;\;\;\;\;\;\;\;\;\;\;\; \cdot \sum\limits_{l = \max \left\{ {0,x - 1} \right\}}^{\min \left\{ {x + 1,N_C  - 1} \right\}} {P\left( {\left. {S_C  = l} \right|S_C  = x} \right)} \nonumber \\
 & \;\;\;\;\;\;\;\;\;\;\;\;\;\;\;  \cdot \sum\limits_{q = 0}^\infty  {P\left( {\left. {Q = q} \right|S_H  = z} \right)} \Phi_y \left( {j,l,q} \right)   \,,
\end{align}where the term $\Phi_y \left( {j,l,q} \right)$, for $y=1,\ldots,N_B-2$, is defined as
\begin{align}
\label{theorem2_proof_new_defined_term}
\Phi_y \left( {j,l,q} \right) &= V_t^1 \left( {j,l,\min \left\{ {N_B  - 1,y + 1 + q} \right\}} \right) \\
&\;\;\;- V_t^0 \left( {j,l,\min \left\{ {N_B  - 1,y + q} \right\}} \right) \nonumber\\
&\;\;\;- V_t^1 \left( {j,l,\min \left\{ {N_B  - 1,y + q} \right\}} \right) \nonumber\\
&\;\;\;+ V_t^0 \left( {j,l,\min \left\{ {N_B  - 1,y - 1 + q} \right\}} \right) \,.\nonumber
\end{align}The third summation over the variable $q$ in (\ref{theorem2_proof_1}) can be further divided into three cases, and after some straightforward manipulations, we obtain
\begin{align}
\label{theorem2_proof_new_defined_term_new_parts}
&\Phi_y \left( {j,l,q} \right) \\
&= \left\{ \begin{array}{l}
 0,\;\;q = 0, \ldots ,\left(N_B  - y - 2\right) \,; \\
  - \left( {V_t^0 \left( {j,l,N_B  - 1} \right) - V_t^0 \left( {j,l,N_B  - 2} \right)} \right) \le 0,\;\;\\
 \;\;\;\;\;\;\;\;\;\;\;\;\;\;\;\;\;\;\;\;\;\;\;\;\;\;\;\;\;\;\;\;\;\;\;\;\;\;\;\;\;\;\;\;\;\;\;\;\; q = \left(N_B  - y - 1\right) \,;\\
 0,\;\;q = \left( {N_B  - y} \right), \ldots ,\infty  \,,
 \end{array} \right. \nonumber
\end{align}
where the inequality in the second line comes from Lemma \ref{lemma1}. From (\ref{theorem2_proof_1}) and (\ref{theorem2_proof_new_defined_term_new_parts}), it leads to $\Lambda _t \left( {z,x,y + 1} \right) \leq \Lambda _t \left( {z, x,y} \right)$, and thus the function $\Lambda _t \left( {z,x,y} \right)$ is non-increasing in $y$. By applying Lemma \ref{lemma2}, it suffices to prove that $\Theta _i \left( {z,x,y} \right)$ is non-decreasing in $y \in \mathcal{B}\backslash \{0\}$. When the value iteration algorithm is converged, a threshold structure ${\boldsymbol \kappa } = \left\{ {{\boldsymbol \kappa }_0 , \ldots ,{\boldsymbol \kappa }_{N_H  - 1} } \right\}$, where ${\boldsymbol \kappa }_z  = \left\{ {\kappa _{z,0} , \ldots ,\kappa _{z,N_C  - 1} } \right\}$, is given by using the non-decreasing property of $\Theta _i \left( {z,x,y} \right)$:
\begin{align}
\label{theorem2_proof_threshold_policy_structure}
\pi ^* \left( {z, x,y} \right) = \left\{ \begin{array}{l}
 0,\;\;y \le \kappa _{z,x} \,; \\
 1,\;\;y \ge \kappa _{z,x}  + 1 \,,
 \end{array} \right.
\end{align}
for a threshold $\kappa _{z, x}$ that is satisfied with $\Theta _i \left( {z,x,\kappa _{z,x} } \right) < 0$ and $\Theta _i \left( {z,x,\kappa _{z, x}  + 1} \right) \ge 0$ if $\kappa _{z,x} \in  \mathcal{B}\backslash \{0\}$, and $\Theta _i \left( {z,x,\kappa _{z,x}  + 1} \right) \ge 0$ if $\kappa _{z, x} = 0$.
\end{proof}

By taking the solar power harvesting model in Table \ref{tab:table_em_training} as an example, a threshold structure is demonstrated in Fig. \ref{Long_term_expected_reward} for the solar state $S_H= 0$. The discount factor and the adopted modulation scheme are respectively set as $\lambda= 0.5$ and 8PSK. It appears that there exists a threshold ${\boldsymbol \kappa }_0  = \left\{ {7, 7, 0, 0, 0, 0 } \right\}$ above which data transmission occurs to gain the maximum long-term expected reward. Furthermore, it can be seen that for a fixed channel state, the long-term expected reward is non-decreasing with respect to the battery state. The simplicity of the threshold structure makes the on-off transmission policy attractive for hardware implementation, and it also helps reduce the computational burden in obtaining the optimal policy.

\begin{figure}[t]
\centering
\includegraphics[width=0.4\textwidth]{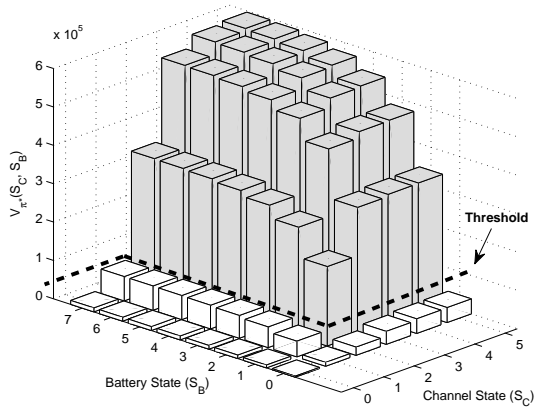}
\caption{Threshold structure policy and long-term expected reward for the solar state $S_H= 0$  ($N_C= 6$, $N_B= 8$, $R_S= 10^5$ symbols/sec, $L_S= 10^3$ symbols/packet, $T_L= 300$ sec, $P_U= 1.8 \times 10^4$ $\mu$W, $\gamma_U= 6$ dB, $\Omega_S= 0.1 $ cm$^2$, $\vartheta= 1$, $f_D= 5 \times 10^{-2}$, and ${\bf \Gamma } = \left\{ {0, 0.3, 0.6, 1.0, 2.0, 3.0, \infty } \right\}$).}
\label{Long_term_expected_reward}
\end{figure}

\subsection{Energy Deficiency Condition}

It is observed from (\ref{Prob_no_quanta}) and (\ref{Prob_no_quanta_jth_state_explict}) that the harvested energy is quantized into two consecutive energy quantum levels, $Q=0$ and $Q=1$, if the harvested power is less than the transmission power $P_U$ (i.e., the mean and variance of each solar state are sufficiently small). The energy level of $Q=0$ is referred to as energy deficiency. A necessary energy deficiency condition for the existence of an optimal threshold policy at ${\boldsymbol \kappa }  = \left\{ {{\boldsymbol \kappa }_0 , \ldots ,{\boldsymbol \kappa }_{N_H  - 1} } \right\}$ is provided in the following.

\begin{theorem} Let $V_{\pi *} \left( {z,x,y} \right)$ be the long-term expected reward of the on-off policy. Define $\Xi \left( {z,x,y} \right) = \mathbb{E}_{z,x} \left[V_{\pi *} \right.$ $\left( {j,l,\min \left\{ {N_B  - 1,y + 1} \right\}} \right) $ $\left.- V_{\pi *} \left( {j,l,\min \left\{ {N_B  - 1,y} \right\}} \right) \right]$ as a difference function of $V_{\pi *} \left( {z,x,y} \right)$ at the two battery states $\min \left\{ {N_B  - 1,y+1} \right\}$ and $\min \left\{ {N_B  - 1,y} \right\}$, which is averaged over the channel and solar state transition probabilities from the state $(z,x) \in \mathcal{H}\times \mathcal{C}$ to its adjacent states. Consider two possible energy quantum levels $Q=0$ and $Q=1$. There exists an optimal policy with the threshold ${\boldsymbol \kappa }  = \left\{ {{\boldsymbol \kappa}_0 , \ldots ,{\boldsymbol \kappa}_{N_H  - 1} } \right\}$, only if the energy deficiency probability belongs to the interval $\mathcal{D}_z= \bigcap\nolimits_{x = 0}^{N_C  - 1} {\mathcal{D}_{z,x } } $, where $\mathcal{D}_{z,x }$ is defined as
\begin{align}
\label{theorem3_polyhedron}
&\mathcal{D}_{z,x}  \\
&= \left\{ \begin{array}{l}
 P\left( {\left. {Q = 0} \right|S_H  = z} \right) \le \phi \left( {z,x,1} \right), \;\; \kappa _{z,x}  = 0 \,;\\
 P\left( {\left. {Q = 0} \right|S_H  = z} \right) \ge \phi \left( {z,x,0} \right), \;\; \kappa _{z,x}  = N_B  - 1 \,;\\
 \phi \left( {z,x,0} \right) \le P\left( {\left. {Q = 0} \right|S_H  = z} \right) \le \phi \left( {z,x,1} \right), \;\; \\
 \;\;\;\;\;\;\;\;\;\;\;\;\;\;\;\;\;\;\;\;\;\;\;\;\;\;\;\;\;\;\;\;\;\;\;\;\;\;\;\;\;\;\;\;\;\;\;\;\;\;\;\;\;\;\;\;\;\;\;\;\;{\rm  otherwise} \,,
 \end{array} \right.\nonumber
\end{align}
where $\phi \left( {z,x,n} \right) = \frac{{{{R_1\left( x \right)} \mathord{\left/
 {\vphantom {{R\left( x \right)} \lambda }} \right.
 \kern-\nulldelimiterspace} \lambda } - \Xi \left( {z,x,\kappa _{z,x}  + n} \right)}}{{\Xi \left( {z,x,\kappa _{z,x}  + n - 1} \right) - \Xi \left( {z,x,\kappa _{z,x}  + n} \right)}}$ and $R_1 \left( x \right) = R_1 \left( {x,\kappa _{z,x}  + 1} \right) = R_1 \left( {x,\kappa _{z,x} } \right)$.
\end{theorem}

\begin{proof}
By applying Theorem \ref{theorem2}, it is sufficient to show that ${\boldsymbol \kappa }$ is the optimal threshold policy, only if the following conditions are satisfied, $\forall z\in \mathcal{H}$ and $\forall x \in \mathcal{C}$:
\begin{align}
\label{theorem3_condition_1}
\left\{ \begin{array}{l}
 V_{\pi *}^1 \left( {z,x,\kappa _{z,x}  + 1} \right) \ge V_{\pi *}^0 \left( {z,x,\kappa_{z,x}  + 1} \right), \;\; \kappa _{z,x}  = 0{\rm  }\,; \\
 V_{\pi *}^1 \left( {z,x,\kappa _{z,x} } \right) \le V_{\pi *}^0 \left( {z,x,\kappa _{z,x} } \right),\;\; \kappa _{z,x}  = N_B  - 1 \,;\\
 V_{\pi *}^1 \left( {z,x,\kappa _{z,x} } \right) \le V_{\pi *}^0 \left( {z,x,\kappa _{z,x} } \right)\;\;{\rm  and }\\
 \;\;\;\;\;\;\;\;\;\; V_{\pi *}^1 \left( {z,x,\kappa _{z,x}  + 1} \right) \ge V_{\pi *}^0 \left( {z,x,\kappa _{z,x}  + 1} \right),\;\; \\
 \;\;\;\;\;\;\;\;\;\;\;\;\;\;\;\;\;\;\;\;\;\;\;\;\;\;\;\;\;\;\;\;\;\;\;\;\;\;\;\;\;\;\;\;\;\;\;\;\;\;\;\;\;\;\;\;\;\;\;\;\;\;\;\;\;{\rm  otherwise} \,.
 \end{array} \right.
\end{align}
From the definition in (\ref{simple_expected_reward_reformulation}), the condition of $V_{\pi *}^1 \left( {z,x,\kappa _{z,x} } \right) \le V_{\pi *}^0 \left( {z,x,\kappa _{z,x} } \right)$ in (\ref{theorem3_condition_1}) becomes
\begin{align}
\label{theorem3_definition_range1}
&R_1\left( x \right) \le \lambda \sum\nolimits_{q = 0}^{1} {P\left( {\left. {Q = q} \right|S_H  = z} \right) \Xi \left( {z,x,\kappa _{z,x}  - 1 + q} \right)},\;\; \nonumber \\
&\;\;\;\;\;\;\;\;\;\;\;\;\;\;\;\;\;\;\;\;\;\;\;\;\;\;\;\;\;\;\;\;\;\;\;\;\;\;\;\;\;\;\;\;\;\;\;\;\; z\in \mathcal{H} \;\;{\rm  and } \;\; x\in \mathcal{C} \,.
\end{align}
On the other hand, the condition of $V_{\pi *}^1 \left( {z,x,\kappa _{z,x}  + 1} \right) \ge V_{\pi *}^0 \left( {z,x,\kappa _{z,x}  + 1} \right)$ implies that
\begin{align}
\label{theorem3_definition_range2}
&R_1\left( x \right) \ge \lambda \sum\nolimits_{q = 0}^{1}  {P\left( {\left. {Q = q} \right|S_H  = z} \right) \Xi \left( {z,x,\kappa _{z,x}  + q} \right)},\;\; \nonumber \\
&\;\;\;\;\;\;\;\;\;\;\;\;\;\;\;\;\;\;\;\;\;\;\;\;\;\;\;\;\;\;\;\;\;\;\;\;\;\;\;\;\;\;\;\;\;\;\;\;\; z\in \mathcal{H} \;\;{\rm  and } \;\; x\in \mathcal{C}\,.
\end{align}
In addition, it can be derived that $\Xi \left( {z,x,\kappa _{z,x}  - 1} \right) - \Xi \left( {z,x,\kappa _{z,x} } \right) \ge 0$ as follows:
\begin{align}
\label{Xi_n0}
&  \Xi \left( {z,x,\kappa _{z,x}  - 1} \right) - \Xi \left( {z,x,\kappa _{z,x} } \right) \\
&  = \mathbb{E}_{z,x} \Big[ V_{\pi *}^0 \left( {j,l,\min \left\{ {N_B  - 1,\kappa _{z,x} } \right\}} \right) \nonumber\\
& \;\;\;\;\;\;\;\;\;\;\;\;\;\;\;\;\;\;\;\;\;\;\;\;\;\;\;\; - V_{\pi *}^0 \left( {j,l,\min \left\{ {N_B  - 1,\kappa _{z,x}  - 1} \right\}} \right) \Big] \nonumber \\
& - \mathbb{E}_{z,x} \Big[ V_{\pi *}^1 \left( {j,l,\min \left\{ {N_B  - 1,\kappa _{z,x}  + 1} \right\}} \right) \nonumber \\
& \;\;\;\;\;\;\;\;\;\;\;\;\;\;\;\;\;\;\;\;\;\;\;\;\;\;\;\; - V_{\pi *}^0 \left( {j,l,\min \left\{ {N_B  - 1,\kappa _{z,x} } \right\}} \right) \Big] \nonumber \\
&  \ge \mathbb{E}_{z,x} \Big[ V_{\pi *}^1 \left( {j,l,\min \left\{ {N_B  - 1,\kappa _{z,x} } \right\}} \right) \nonumber \\
& \;\;\;\;\;\;\;\;\;\;\;\;\;\;\;\;\;\;\;\;\;\;\;\;\;\;\;\;  - V_{\pi *}^0 \left( {j,l,\min \left\{ {N_B  - 1,\kappa _{z,x}  - 1} \right\}} \right) \Big] \nonumber \\
& - \mathbb{E}_{z,x} \Big[ V_{\pi *}^1 \left( {j,l,\min \left\{ {N_B  - 1,\kappa _{z,x}  + 1} \right\}} \right) \nonumber\\
& \;\;\;\;\;\;\;\;\;\;\;\;\;\;\;\;\;\;\;\;\;\;\;\;\;\;\;\;  - V_{\pi *}^0 \left( {j,l,\min \left\{ {N_B  - 1,\kappa _{z,x} } \right\}} \right) \Big] \ge 0 \,, \nonumber
\end{align}
where the last inequality holds due to (\ref{theorem2_proof_new_defined_term}) and (\ref{theorem2_proof_new_defined_term_new_parts}). Similarly, we get
\begin{align}
\label{Xi_n1}
\Xi \left( {z,x,\kappa _{z,x} } \right) - \Xi \left( {z,x,\kappa _{z,x}  + 1} \right) \ge 0 \,.
\end{align}
By applying (\ref{theorem3_definition_range1})-(\ref{Xi_n1}) into (\ref{theorem3_condition_1}) and using $P\left( {\left. {Q = 0} \right|S_H  = z} \right)$ $ + P\left( {\left. {Q = 1} \right|S_H  = z} \right)=1$, the necessary conditions can be rewritten as in (\ref{theorem3_polyhedron}). It is concluded that there exists an optimal threshold at ${\boldsymbol \kappa }  = \left\{ {{\boldsymbol \kappa }_0 , \ldots ,{\boldsymbol \kappa }_{N_H  - 1} } \right\}$, only if the probability $P\left( {\left. {Q = 0} \right|S_H  = z} \right)$ $ \in \mathcal{D}_z= \bigcap\nolimits_{x = 0}^{N_C  - 1} {\mathcal{D}_{z,x } } $.

\end{proof}

This necessary condition gives an important insight into how the energy deficiency probability affects the threshold of the policy. Taking the long-term expected reward in Fig. \ref{Long_term_expected_reward} and $S_H=0$ as an example, the energy deficiency regions versus the immediate rewards $R_1(x=2)$ for different thresholds $\kappa_{0,2}$ are plotted in Fig. \ref{Energy_Deficiency_Reward}, where the other thresholds are fixed at $\left\{ {\kappa _{0,0} ,\kappa _{0,1} ,\kappa _{0,3} ,\kappa _{0,4} ,\kappa _{0,5} } \right\} = \left\{ {7,7,0,0,0} \right\}$. It can be observed that for $R_1(x=2)= 2\times 10^4$ and $6\times 10^4$, the threshold $\kappa_{0,2}= 1 $ could be the optimal policy, only if the energy deficiency probability $P\left( {\left. {Q = 0} \right|S_H  = 0} \right)\le 0.25$ and $P\left( {\left. {Q = 0} \right|S_H  = 0} \right)\ge 0.5$, respectively.

\begin{figure}[t]
\centering
\includegraphics[width=0.4\textwidth]{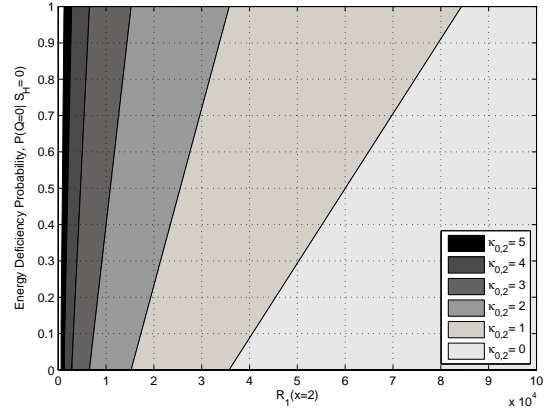}
\caption{Energy deficiency regions $P\left( {\left. {Q = 0} \right|S_H  = 0} \right)$ versus immediate rewards $R_1(x=2)$ for different thresholds $\kappa_{0,2}$.}
\label{Energy_Deficiency_Reward}
\end{figure}

\subsection{Expected Net Bit Rate Analysis}
Here we use the expected net bit rate to assess the performance of the optimal threshold policy. Consider a threshold policy ${\boldsymbol \kappa } = \left\{ {{\boldsymbol \kappa }_0 , \ldots ,{\boldsymbol \kappa }_{N_H  - 1} } \right\}$, and denote $\nu_{j,i\times N_B+n}$ as the stationary probability of the state $\left( {S_H, S_C ,S_B } \right) = \left( {j,i,n} \right)$, for $i=0,\ldots,N_C-1$ and $n=0,\ldots,N_B-1$. Define ${\boldsymbol \nu }_j  = \left[ {\nu_{j,0} , \ldots , \nu_{j,i\times N_B+n}, \ldots,\nu_{j,N_C\times N_B  - 1} } \right]^T $, for $j=0,\ldots,N_H-1$, and ${\boldsymbol \nu } = \left[ {{\boldsymbol \nu }_0^T , \ldots ,{\boldsymbol \nu }_{N_H  - 1}^T } \right]^T $. Let ${{\bf \Pi}_{j,i} }$ be an $N_B \times N_B$ battery state transition probability matrix associated with the threshold policy ${\boldsymbol \kappa }$ at the $j^{th}$ solar state and the $i^{th}$ channel state, given by
\begin{align}
\label{battery_state_transition_prob_performance_analysis}
& \left[ {{\bf \Pi }_{j,i} } \right]_{p,q}  \\
& = \left\{ \begin{array}{l}
 P\left( {\left. {Q = \left( {p - q} \right)} \right|S_H  = j} \right), \\
 \;\;\;\;\;\;\;\;\;0 \le q \le \kappa _{j,i} , \;\; q \le p \le N_B  - 2 \,;\\
 0,\;\; 0 \le q \le \kappa _{j,i} ,\;\; 0 \le p \le q - 1  \,;\\
 P\left( {\left. {Q = \left( {p - q + 1} \right)} \right|S_H  = j} \right), \\
 \;\;\;\;\;\;\;\;\; \kappa _{j,i}  + 1 \le q \le N_B  - 1,\;\;q - 1 \le p \le N_B  - 2 \,; \\
 0,\;\; \kappa _{j,i}  + 1 \le q \le N_B  - 1, \;\; 0 \le p \le q - 2 \,,
 \end{array} \right. \nonumber
\end{align}
and $\left[ {{\bf \Pi }_{j,i} } \right]_{N_B  - 1,q}  = 1 - \sum\nolimits_{p = 0}^{N_B  - 2} {\left[ {{\bf \Pi }_{j,i} } \right]_{p,q} } $, for $q= 0, \ldots, N_B-1$, where the $(p,q)^{th}$ entry of the matrix $\left[ {{\bf \Pi }_{j,i} } \right]$ represents the transition probability from the state $\left( {S_H, S_C ,S_B } \right) = \left( {j,i,q} \right)$ to the state $\left( {S_H, S_C ,S_B } \right) = \left( {j,i,p} \right)$. Therefore, the stationary probability with respect to the threshold policy ${\boldsymbol \kappa }$ can be computed by solving the balance equation:
\begin{align}
\label{balance_equation_performance_analysis}
\left[ {\begin{array}{*{20}c}
   {{\bf \Phi } - {\bf I}_{\left( {N_B  \times N_C \times N_H} \right)} }  \\
   {{\bf 1}_{\left( {N_B  \times N_C \times N_H} \right)}^T }  \\
\end{array}} \right]{\boldsymbol \nu } =  \left[ \begin{array}{l}
 {\bf 0}_{\left( {N_B  \times N_C \times N_H} \right)}  \\
 1 \\
 \end{array} \right]\,,
\end{align}where ${\bf \Phi }$ is the state transition probability matrix of size $\left( {N_B  \times N_C \times N_H} \right) \times \left( {N_B  \times N_C \times N_H} \right)$, whose $(z N_C+x, j N_C+i)^{th}$ sub-matrix is equal to $P\left( {\left. {S_H  = z} \right|S_H  = j} \right) \cdot P\left( {\left. {S_C  = x} \right|S_C  = i} \right) \cdot {\bf \Pi}_{j,i}$, for $z, j=0,\ldots,N_H-1$, $i=0,\ldots,N_C-1$, and $x=\max\{0,i-1\},\ldots,\min\{i+1,N_C-1\}$, and the remaining sub-matrices all equate to zero. By taking the expectation of the reward function in (\ref{Reward_function_good_bits}), the expected net bit rate using the $2^{\chi_m}$-ary modulation scheme is given by
\begin{align}
\label{net_bit_rate_performance_analysis}
 R_{net,m} &= \frac{1}{{T_P }}\sum\limits_{j = 0}^{N_H  - 1} \sum\limits_{i = 0}^{N_C  - 1} \sum\limits_{n \ge \kappa _{j,i}  + 1}^{N_B  - 1} {\nu _{j,\left( {i \times N_B  + n} \right)}}\nonumber\\
 &\;\;\;\;\;\; \cdot { { \chi _m L_S \left( {1 - \eta \left( {i,n,1,m} \right)} \right)^{\chi _m L_S } } }    \,.
\end{align}

\begin{theorem}
\label{theorem4}
Define an energy harvesting rate as $\bar q = \mathop {\lim }\limits_{T \to \infty } \bar q_T  = \mathop {\lim }\limits_{T \to \infty } \mathbb{E}\left[ {\frac{1}{T}\sum\nolimits_{t = 1}^T {q_t } } \right]$, where $q_t$ denotes the number of energy quanta obtained by a sensor at the $t^{th}$ policy management period. The expected net bit rate of the on-off policy is upper bounded by
\begin{align}
\label{theorem4_upper_bound}
&R_{net,m}  \le \min \left\{ {\bar q,1} \right\} \cdot \Big( \frac{1}{{T_P }}\chi _m L_S \nonumber \\
&\;\;\;\;\;\;\;\;\;\;\;\;\;\;\;\;\;\; \cdot \left( {1 - \eta \left( {N_C  - 1,N_B  - 1,1,m} \right)} \right)^{\chi _m L_S }  \Big) \,.
\end{align}
At asymptotically high SNR, the upper bound value converges to $\min \left\{ {\bar q,1} \right\}\cdot \frac{1}{{T_P }}\chi _m L_S $.
\end{theorem}

\begin{proof}
Let $a_t \in \{0, 1\}$ be the optimal action at the $t^{th}$ policy management period, corresponding to a sequence of channel states $x_t$ and battery states $y_t$, for $t=1,\ldots,T$. From (\ref{Reward_function_good_bits}), the immediate reward can be rewritten as $R_m \left( {a_t ,x_t ,y_t } \right) = a_t \frac{1}{{T_P }}\chi _m L_S \left( {1 - \eta \left( {x_t ,y_t ,1,m} \right)} \right)^{\chi _m L_S } $. Thus, the average net bit rate is calculated as
\begin{align}
\label{theorem4_average_net_bit_rate_definition}
& R_{net,m}  = \mathop {\lim }\limits_{T \to \infty } \mathbb{E}\left[ {\frac{1}{T}\sum\limits_{t = 1}^T {R_m \left( {a_t ,x_t ,y_t } \right)} } \right] \\
& = \mathop {\lim }\limits_{T \to \infty } \sum\limits_{i_t } {P\left( {x_t  = i_t ,t = 1, \ldots ,T} \right)}  \nonumber\\
& \;\;\; \cdot \frac{1}{T}\sum\limits_{t = 1}^T  \mathbb{E}\Big[ \left. {a_t \frac{1}{{T_P }}\chi _m L_S \left( {1 - \eta \left( {x_t ,y_t ,1,m} \right)} \right)^{\chi _m L_S } } \right| \nonumber \\
& \;\;\;\;\;\;\;\;\;\;\;\;\;\;\;\;\;\;\;\;\;\;\;\;\;\;\;\;\;\;\;\;\;\;\;\;\;\;\;\;\;\;\;\;\;\;\;\;\;\;\;\;\;\; x_t  = i_t ,t = 1, \ldots ,T \Big]   \nonumber\\
& \le \mathop {\lim }\limits_{T \to \infty } \sum\limits_{i_t } P\left( {x_t  = i_t ,t = 1, \ldots ,T} \right) \nonumber\\ & \;\;\; \cdot \frac{1}{T}\sum\limits_{t = 1}^T {\mathbb{E}\Big[ {\left. {a_t } \right|x_t  = i_t ,t = 1, \ldots ,T} \Big]} \nonumber\\
& \;\;\; \cdot {\left( {\frac{1}{{T_P }}\chi _m L_S \left( {1 - \eta \left( {N_C  - 1,N_B  - 1,1,m} \right)} \right)^{\chi _m L_S } } \right)}\,. \nonumber
\end{align}
For any transmission policy, the accumulated energy consumption cannot exceed the initial energy in the battery plus the total amount of harvested energy, and it yields the following constraint:
\begin{align}
\label{theorem4_energy_neurality_constraint}
\frac{1}{T}\sum\nolimits_{t = 1}^T {a_t }  \le \frac{1}{T}\left( {N_B  - 1} \right) + \frac{1}{T}\sum\nolimits_{t = 1}^T {q_t } \,.
\end{align}
Besides, the on-off transmission imposes another energy expenditure constraint:
\begin{align}
\label{theorem4_energy_onoff_constraint}
\frac{1}{T}\sum\nolimits_{t = 1}^T {a_t }  \le 1 \,.
\end{align}
By applying (\ref{theorem4_energy_neurality_constraint}) and (\ref{theorem4_energy_onoff_constraint}) into (\ref{theorem4_average_net_bit_rate_definition}), it gives
\begin{align}
\label{theorem4_average_net_bit_rate_derivation1}
& R_{net,m}  \le \min \Big\{ {\mathop {\lim }\limits_{T \to \infty } \left( {\frac{{N_B  - 1}}{T} + \bar q_T } \right),1} \Big\}\nonumber\\
& \;\;\;\cdot \left( {\frac{1}{{T_P }}\chi _m L_S \left( {1 - \eta \left( {N_C  - 1,N_B  - 1,1,m} \right)} \right)^{\chi _m L_S } } \right) \nonumber \\
&  = \min \left\{ {\bar q,1} \right\} \nonumber \\
& \;\;\;\cdot \frac{1}{{T_P }}\chi _m L_S \left( {1 - \eta \left( {N_C  - 1,N_B  - 1,1,m} \right)} \right)^{\chi _m L_S } \,.
\end{align}
Finally, it is obtained from (\ref{Reward_function_BER}) that the function ${\eta \left( {N_C  - 1,N_B  - 1,1,m} \right)} \rightarrow 0 $ as $\gamma_U \rightarrow \infty$, and the upper bound converges to $\min \left\{ {\bar q,1} \right\}\cdot \frac{1}{{T_P }}\chi _m L_S $ at asymptotically high SNR.

\end{proof}


\section{Simulation Results}
Simulation results are presented in this section to evaluate the performance of the proposed data-driven transmission policies. In the system model, the numbers of solar states, battery states, channel states are set as four, twelve, and six, respectively. For convenience, the data record of the irradiance collected by the solar site in Elizabeth City State University in June from 2008 to 2012 is adopted throughout the simulation \cite{N.R.E.Laboratory12}. A four-state solar power harvesting model is trained using the data in 2008, 2009 and 2010, where the underlying parameters are given in Table I. The irradiance data of the subsequent two years, 2011 and 2012, are then applied for performance evaluation. Other simulation parameters are listed in Table II. The channel quantization levels are defined as ${\bf \Gamma } = \left\{ {0, 0.3, 0.6, 1.0, 2.0, 3.0, \infty } \right\}$, and the channel gains are generated by Jakes' model with the normalized Doppler frequency $f_D= 0.05$ \cite{W.C.Jakes74}. In the system configuration, each packet contains $L_S= 10^3$ data symbols, and the symbol rate $R_S$ is operated at $100$ kHz. In other words, the packet duration $T_P$ is given by $0.01$ sec. The modulation types could be QPSK, 8PSK and 16QAM, and the basic transmission power level is chosen as $P_U= 40\times 10^3$ $\mu$W. These three modulation types are considered as potential candidates for the composite policy, while only one modulation type is preselected for the on-off policy. The transmission actions are changed every five minutes, i.e., $T_L= 300$ sec. Different actions are followed by different modulation and power choices, resulting in different bit rate performance. In the value iteration algorithm, the discount factor $\lambda$ and the stopping criterion $\varepsilon$ are selected as $0.99$ and $10^{-6}$, respectively. The solar panel area is assumed to be $1$ cm$^2$, $4$ cm$^2$ and $8$ cm$^2$, and the energy conversion efficiency is set as $\vartheta= 20 \%$ \cite{S.Sudevalayam11}. We assume that the battery state is randomly initialized. The above parameters are used as default settings, except as otherwise stated. Finally, a normalized SNR $\gamma_C$ is defined with respect to the transmission power of $10^3$ $\mu$W throughout the simulation.

\begin{table}[t]
\renewcommand{\arraystretch}{1.2}\tabcolsep=0.6ex
\caption{Simulation parameters}
\label{Simulation_Parameters} \centering
\begin{tabular}{cc}
\hline\hline
Symbol rate ($R_S$) & $100$ kHz \\
\hline
Packet size ($L_S$) & $10^3$ symbols \\
\hline
\multirow{3}{*}{Modulation type ($\alpha_m, \beta_m$)} & QPSK: ($1,2$)\\
 & 8PSK: ($\frac{2}{3}, 2\sin ^2 \left( {\frac{\pi }{8}} \right)$)\\
 & 16QAM: ($\frac{3}{4}, \frac{3}{15}$) \\
\hline
Policy management duration ($T_L$) & $300$ sec \\
\hline
Basic action power ($P_U$) & $40 \times 10^3$ $\mu$W \\
\hline
Solar panel area ($\Omega_S$) & $1$, $4$, and $8$ cm$^2$ \\
\hline
Energy conversion efficiency ($\vartheta$) &  $20 \%$ \\
\hline
Channel quantization levels ($ \bf \Gamma $) & $\left\{ {0, 0.3, 0.6, 1.0, 2.0, 3.0, \infty } \right\}$ \\
\hline
Channel Model & Jakes' model \\
\hline
Normalized Doppler frequency ($f_D$) &  $0.05$ and $0.005$ \\
\hline
Discount factor ($\lambda$) &  $0.99$ \\
\hline
\hline
\end{tabular}
\label{System_parameter}
\end{table}

As a benchmark, two myopic policies are included for performance comparisons. For these two policies, the actions are performed without concern for the channel state and battery state transition probabilities, and data packets are transmitted as long as the battery storage is non-empty. The first policy (Myopic Policy I) attempts to transmit data packets at the lowest transmission power level, if the energy storage is positive. Regarding with the second one (Myopic Policy II), the largest available battery power is consumed for data transmission, if the battery state is non-zero. In addition, we compare the proposed schemes with a deterministic energy harvesting scheme in \cite{M.Gorlatova13}, called $t$-time fair rate assignment ($t$-TFR), which requires perfect knowledge of the channel fading and energy harvesting patterns for determining the optimal transmission power over a short-term period $t$ in order to maximize the reward function in (\ref{Reward_function_good_bits}).

Fig. \ref{New_Simu_SNR_Net_Bit_Sollarcell} shows the expected net bit rates for the composite and on-off transmission policies. The solar panel area is set as $\Omega_S= 1$ cm$^2$. The expected net bit rate of the on-off policy is calculated according to (\ref{net_bit_rate_performance_analysis}), while that for the composite policy can be analyzed in a similar way although the accessible transmission actions appear to be more sophisticated. The performance upper bound of the on-off policy in (\ref{theorem4_upper_bound}) is also included for calibration purposes. For the on-off policy, it is observed that the expected net bit rate is monotonically increased with the operating SNRs, while the performance finally becomes saturated at $0.6 \times 10^5$ bits/sec, $0.9 \times 10^5$ bits/sec and $1.2 \times 10^5$ bits/sec for QPSK, 8PSK and 16QAM, respectively, when $\gamma_C$ is sufficiently high. It is clear that the policy with QPSK modulation exhibits a better bit rate, as compared to 8PSK and 16QAM modulation when $\gamma_C \le 2$ dB. On the contrary, it is advisable to employ high-level modulation schemes, e.g., 8PSK and 16QAM, to achieve better performance. This is because the adoption of high-level modulation schemes generally requires larger SNRs in order to guarantee a low packet error rate. As expected, the composite policy offers an expected net bit rate better than the on-off policy, and the performance gap between these two policies could be as large as $60 \times 10^3$ bits/sec. However, the on-off policy with a mixture of QPSK and 16QAM modulation can still achieve a large fraction of bit rate regions as available in the composite policy, and its simple implementation makes it attractive for practical applications.

Fig. \ref{New_Simu_Observe_Net_Bit_Rate_Composite_Policies} shows the average net bit rates of the proposed composite policy and other benchmark schemes, in which the real data record from 2011 and 2012 is utilized to assess the performance. Here, the solar panel area is set to be $\Omega_S= 1$ cm$^2$. We can observe from this figure that Myopic Policy I with QPSK is superior to Myopic Policy II with 16QAM in terms of the average net bit rates for low SNR regions, whereas the reverse trend is found for high SNR regions. This is because aggressive energy expenditure merits better bit rate performance when the operating SNR is high, and conservative use of energy is more preferable at low SNRs. In agreement with the theoretical results in Fig. \ref{New_Simu_SNR_Net_Bit_Sollarcell}, the composite transmission policy is capable of achieving much better average net bit rates than these two myopic policies when real data measurement is used. We can also find that the average net bit rate of the composite policy is superior to that of the $t$-TFR scheme, even if the energy harvesting and channel variation patterns are assumed to be perfectly predicted for one or two hours. Though the $t$-TFR scheme could attain better performance with an increased prediction interval, it suffers from the problems of larger prediction error and higher computational complexity for a long prediction interval.

The average net bit rate of the on-off transmission policy is shown in Fig. \ref{New_Simu_Observe_Net_Bit_Rate_Policies_ONOFF} for different modulation types, where the solar panel area is set as $\Omega_S= 1$ cm$^2$. Moreover, the performances of the Myopic Policy I and the Two Hour-TFR schemes, in conjunction with various modulation types, are included in this figure. In order to make a fair comparison, the $t$-TFR scheme also adopts on-off power actions for the short-term scheduling of energy expenditure. It can be seen that the maximum spectrum efficiency provided by our proposed on-off policy is approximately given by $0.6$ bits/sec/Hz and $1.2$ bits/sec/Hz for QPSK and 16QAM, respectively. With a fixed modulation scheme, the on-off policy offers significant performance gains over the myopic policy by taking advantage of channel diversity gains. A closer look at this figure reveals that the performance gap between these two policies becomes wider as the modulation level increases. When compared with the Two Hour-TFR scheme, the on-off policy can still achieve better average net bit rates, no matter which modulation type is used.

Fig. \ref{New_Simu_Observe_Net_Bit_Rate_Battery_State} illustrates the average net bit rate of the composite policy as a function of the number of battery states. To clearly understand the relationship between the Doppler frequency and the battery storage capacity, the normalized Doppler frequency, $f_D$, is chosen as $0.005$ and $0.05$. We can observe that the average net bit rate can be dramatically enhanced by enlarging the energy buffer size to store more energy quanta, especially when the operating SNR is low. For instance, the performance with $N_B= 16$ at $\gamma_C= 0$ dB, $\Omega_S= 8$ cm$^2$ and $f_D= 0.05$ is about $2.5 \times 10^{5}$ bits/sec, probably $1.5$ times that being achieved by the same policy with $N_B= 2$. Obviously, the energy harvesting sensor node additionally benefits from channel diversity gains if the energy spending is carefully governed to respond to the change in channel conditions. Furthermore, the bit rate becomes better as the Doppler frequency and the solar panel area increase, and the improvement owing to the increase in the number of battery states is relatively modest for lower Doppler frequencies.

\section{Conclusions}

In this paper, we have studied the problem of maximizing long-term net bit rates in sensor communication that solely relies on solar energy for data transmission. A node-specific energy harvesting model was developed to classify the harvesting conditions into several solar states with different energy quantum arrivals. Unlike previous works, which were not concerned with the real-world energy harvesting capability, a data-driven MDP framework was formulated to obtain the optimal transmission parameters from a set of power and modulation actions in response to the dynamics of channel fading and battery storage. Since different nodes may possess different energy harvesting capabilities, the parameters of the underlying energy harvesting process were completely determined by the solar irradiance observed at a sensor node. In practice, the exact solar state at each time epoch is unavailable, and a mixed strategy was proposed to associate the adaptive transmission parameters with the beliefs of the solar states. The validity of the proposed data-driven approach was rigorously justified by the real data of solar irradiance. We also analyzed the properties and the net bit rates of the optimal on-off transmission policy, and it was proved that this policy has an inherent threshold structure in the direction along the battery states. Through extensive computer simulations, the proposed data-driven approach was shown to achieve significant gains with respect to other radical approaches, while it did not require non-causal knowledge of energy harvesting and channel fading patterns.

\begin{figure}[t]
\centering
\includegraphics[width=0.4\textwidth]{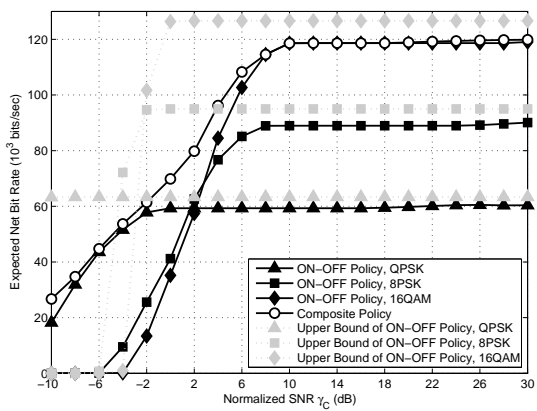}
\caption{Expected net bit rate versus normalized SNR $\gamma_C$ for different transmission policies ($\Omega_S= 1$ cm$^2$, and $f_D= 0.05$).}
\label{New_Simu_SNR_Net_Bit_Sollarcell}
\end{figure}

\begin{figure}[t]
\centering
\includegraphics[width=0.4\textwidth]{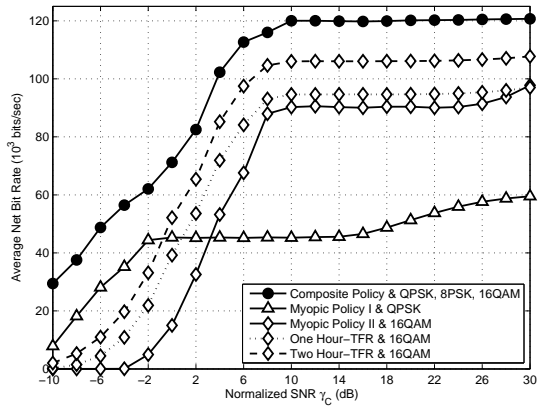}
\caption{Average net bit rate performances of the composite policy, Myopic Policy I, Myopic Policy II and $t$-TFR with the real data record of irradiance in June from 2011 to 2012, measured by a solar site in Elizabeth City State University ($\Omega_S= 1$ cm$^2$, and $f_D= 0.05$).}
\label{New_Simu_Observe_Net_Bit_Rate_Composite_Policies}
\end{figure}

\begin{figure}[t]
\centering
\includegraphics[width=0.4\textwidth]{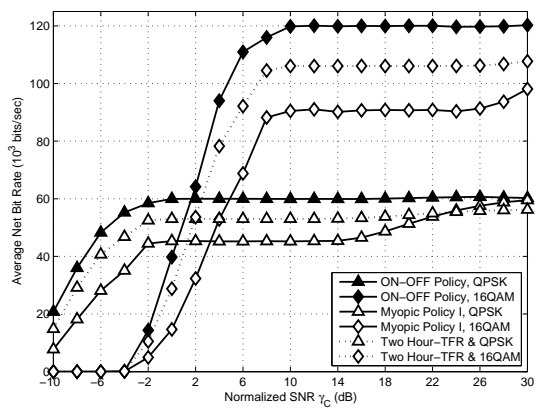}
\caption{Average net bit rate performances of the on-off and other benchmark policies ($\Omega_S= 1$ cm$^2$, and $f_D= 0.05$).}
\label{New_Simu_Observe_Net_Bit_Rate_Policies_ONOFF}
\end{figure}

\begin{figure}[t]
\centering
\includegraphics[width=0.4\textwidth]{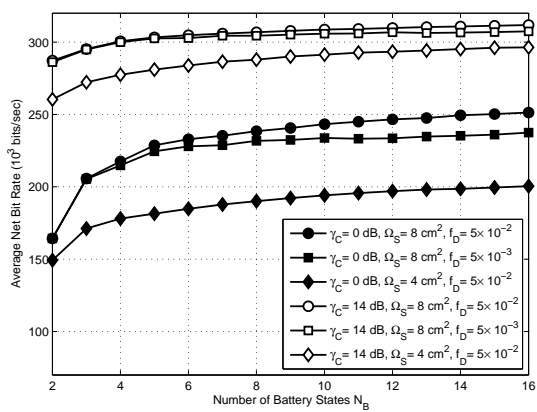}
\caption{Average net bit rate of the composite policy versus number of battery states under different Doppler frequencies and solar panel areas.}
\label{New_Simu_Observe_Net_Bit_Rate_Battery_State}
\end{figure}


\begin{thebibliography}{1}

\bibitem{C.Pandana05} C. Pandana and K. J. R. Liu,
                ``Near optimal reinforcement learning framework for energy-aware wireless sensor communications, "
                {\it IEEE J. Sel. Areas Commun.}, vol. 23, no. 4, pp. 788-797, Apr. 2005.

\bibitem{S.Sudevalayam11} S. Sudevalayam and P. Kulkarni,
                ``Energy harvesting sensor nodes: survey and implications, "
                {\it IEEE Commun. Surveys Tutorials}, vol. 13, no. 3, pp. 443-461, Third Quad. 2011.

\bibitem{A.Kansal07} A. Kansal, J. Hsu, S. Zahedi, and M. B. Srivastava,
                ``Power management in energy harvesting sensor networks, "
                {\it ACM Trans. Embedded Comput. Syst.}, vol. 6, no. 4, pp. 32/1-38, Sept. 2007.

\bibitem{M.Tacca07} M. Tacca, P. Monti, and A. Fumagalli,
                ``Cooperative and reliable ARQ protocols for energy harvesting wireless sensor nodes, "
                {\it IEEE Trans. Wireless Commun.}, vol. 6, no. 7, pp. 2519-2529, July 2007.

\bibitem{S.Reddy10} S. Reddy and C. R. Murthy,
                ``Profile-based load scheduling in wireless energy harvesting sensors for data rate maximization, "
                {\it Proc. IEEE Int. Conf. Commun.}, pp. 1-5, 2010.

\bibitem{D.Niyato07} D. Niyato, E. Hossain, and A. Fallahi,
                ``Sleep and wakeup strategies in solar-powered wireless sensor/mesh networks: performance analysis and optimization, "
                {\it IEEE Trans. Mobile Comput.}, vol. 6, no. 2, pp. 221-236, Feb. 2007.

\bibitem{B.Medepally09} B. Medepally, N. B. Mehta, and C. R. Murthy,
                ``Implications of energy profile and storage on energy harvesting sensor link performance, "
                {\it Proc. IEEE Glob. Commun. Conf.}, pp. 1-6, 2009.

\bibitem{N.Michelusi12} N. Michelusi, K. Stamatiou, and M. Zorzi,
                ``On optimal transmission policies for energy harvesting devices, "
                {\it Proc. IEEE Inf. Theory and App. Workshop}, pp. 249-254, 2012.

\bibitem{N.Michelusi131} N. Michelusi and M. Zorzi,
                ``Optimal random multiaccess in energy harvesting wireless sensor networks, "
                {\it Proc. IEEE Int. Conf. Commun.}, pp. 463-468, 2013.

\bibitem{A.Aprem13} A. Aprem, C. R. Murthy, and N. B. Mehta,
                ``Transmit power control policies for energy harvesting sensors with retransmissions, "
                {\it IEEE J. Sel. Topics Signal Process.}, vol. 7, no. 5, pp. 895-906, Oct. 2013.

\bibitem{K.J.Prabuchandran13} K. J. Prabuchandran, S. K. Meena, and S. Bhatnagar,
                ``Q-learning based energy management policies for a single sensor node with finite buffer, "
                {\it IEEE Wireless Commun. Lett.}, vol. 2, no. 1, pp. 82-85, Feb. 2013.

\bibitem{J.Lei09} J. Lei, R. Yates, and L. Greenstein,
                ``A generic model for optimizing single-hop transmission policy of replenishable sensors, "
                {\it IEEE Trans. Wireless Commun.}, vol. 8, no. 2, pp. 547-551, Feb. 2009.

\bibitem{S.Mao12} S. Mao, M. H. Cheung, and V. W. S. Wong,
                ``An optimal energy allocation algorithm for energy harvesting wireless sensor networks, "
                {\it Proc. IEEE Int. Conf. Commun.}, pp. 265-270, 2012.

\bibitem{M.Kashef12} M. Kashef and A. Ephremides,
                ``Optimal packet scheduling for energy harvesting sources on time varying wireless channels, "
                {\it J. Commun. and Networks}, vol. 14, no. 2, pp. 121-129, Apr. 2012.

\bibitem{Z.Wang12} Z. Wang, A. Tajer, and X. Wang,
                ``Communication of energy harvesting tags, "
                {\it IEEE Trans. Commun.}, vol. 60, no. 4, pp. 1159-1166, Apr. 2012.

\bibitem{H.Li10} H. Li, N. Jaggi, and B. Sikdar,
                ``Cooperative relay scheduling under partial state information in energy harvesting sensor networks, "
                {\it Proc. IEEE Glob. Commun. Conf.}, pp. 1-5, 2010.

\bibitem{N.Michelusi13} N. Michelusi, K. Stamatiou, and M. Zorzi,
                ``Transmission policies for energy harvesting sensors with time-correlated energy supply, "
                {\it IEEE Trans. Commun.}, vol. 61, no. 7, pp. 2988-3001, July 2013.

\bibitem{S.Yin13} S. Yin, E. Zhang, J. Li, L. Yin, and S. Li
                ``Throughput optimization for self-powered wireless communications with variable energy harvesting rate, "
                {\it Proc. IEEE Wireless Commun. and Networking Conf.}, pp. 830-835, 2013.

\bibitem{P.S.Khairnar11} P. S. Khairnar and N. B. Mehta,
                ``Power and discrete rate adaptation for energy harvesting wireless nodes, "
                {\it Proc. IEEE Int. Conf. Commun.}, pp. 1-5, 2011.

\bibitem{O.Ozel11} O. Ozel, K. Tutuncuoglu, J. Yang, S. Ulukus, and A. Yener,
                ``Transmission with energy harvesting nodes in fading wireless channels: optimal policies, "
                {\it IEEE J. Sel. Areas Commun.}, vol. 29, no. 8, pp. 1732-1743, Sept. 2011.

\bibitem{N.Roseveare14} N. Roseveare and B. Natarajan,
                ``An alternative perspective on utility maximization in energy-harvesting wireless sensor networks, "
                {\it IEEE Trans. Veh. Technol.}, vol. 63, no. 1, pp. 344-356, Jan. 2014.

\bibitem{T.Zhang13} T. Zhang, W. Chen, Z. Han, and Z. Cao,
                ``A cross-layer perspective on energy harvesting aided green communications over fading channels, "
                {\it in Proc. IEEE INFOCOM.}, pp. 3225-3230, 2013.

\bibitem{A.Seyedi10} A. Seyedi and B. Sikdar,
                ``Energy efficient transmission strategies for body sensor networks with energy harvesting, "
                {\it IEEE Trans. Commun.}, vol. 58, no. 7, pp. 2116-2126, July 2010.

\bibitem{S.Zhang13} S. Zhang, A. Seyedi, and B. Sikdar,
                ``An analytical approach to the design of energy harvesting wireless sensor nodes, "
                {\it IEEE Trans. Wireless Commun.}, vol. 12, no. 8, pp. 4010-4024, Aug. 2013.

\bibitem{M.Gorlatova13} M. Gorlatova, A. Wallwater, and G. Zussman,
                ``Networking low-power energy harvesting devices: measurements and algorithms, "
                {\it IEEE Trans. Mobile Comput.}, vol. 12, no. 9, pp. 1853-1865, Sept. 2013.

\bibitem{Q.Wang13} Q. Wang and M. Liu,
                ``When simplicity meets optimality: efficient transmission power control with stochastic energy harvesting, "
                {\it Proc. IEEE INFOCOM.}, pp. 580-584, 2013.

\bibitem{N.R.E.Laboratory12} N. R. E. Laboratory. (2012, Feb.)
                Solar radiation resource information. [Online]. Available: http://www.nrel.gov/rredc/.

\bibitem{J.A.Bilmes98} J. A. Bilmes,
                ``A gentle tutorial of the EM algorithm and its application to parameter estimation for Gaussian mixture and hidden Markov models, "
                International Computer Science Institute, Tech. Rep. ICSI-TR-97-021, Apr. 1998.

\bibitem{H.S.Wang95} H. S. Wang and N. Moayeri,
                ``Finite-state Markov channel-a useful model for radio communication channels, "
                {\it IEEE Trans. Veh. Technol.}, vol. 44, no. 1, pp. 163-171, Feb. 1995.

\bibitem{W.C.Jakes74} W. C. Jakes, {\it Microwave Mobile Communications}. New York: Wiley, 1974.


\end{thebibliography}
\end{document}